\documentclass[conference]{IEEEtran}
\IEEEoverridecommandlockouts
\usepackage{cite}[nobreak]
\usepackage{amsmath,amssymb,amsfonts,amsthm}
\usepackage{latexsym}
\usepackage{graphicx}
\usepackage{textcomp}
\usepackage[dvipsnames]{xcolor}

\usepackage{mathrsfs}

\usepackage{balance}  

\usepackage{algorithmicx}
\usepackage{algorithm}
\usepackage[noend]{algpseudocode}

\renewcommand{\algorithmicrequire}{\textbf{Input: }}

\algnewcommand{\algorithmicand}{\textbf{ and }}
\algnewcommand{\algorithmicor}{\textbf{ or }}
\algnewcommand{\Or}{\algorithmicor}
\algnewcommand{\And}{\algorithmicand}
\algnewcommand{\var}{\texttt}

\usepackage{subcaption}
\usepackage{float}

\usepackage{url}
\usepackage{hyperref}

\usepackage{pifont}

\usepackage{mdframed}
\usepackage{times}
\usepackage{footmisc}

\usepackage[T1]{fontenc}
\usepackage[scaled=0.7]{DejaVuSansMono}
\newtheorem{lemma}{Lemma}

\def\BibTeX{{\rm B\kern-.05em{\sc i\kern-.025em b}\kern-.08em
    T\kern-.1667em\lower.7ex\hbox{E}\kern-.125emX}}
\begin{document}

\title{GX-Plug: a Middleware for Plugging Accelerators to Distributed Graph Processing}

\author{Kai Zou, Xike Xie$^*$, Qi Li, Deyu Kong\\
Data Darkness Lab, University of Science and Technology of China\\
\{slnt, likamo, cavegf\}@mail.ustc.edu.cn, $^*$xkxie@ustc.edu.cn
}

\maketitle

\begin{abstract}
Recently, research communities highlight the necessity of formulating a scalability continuum for large-scale graph processing, which gains the scale-out benefits from distributed graph systems, and the scale-up benefits from high-performance accelerators. To this end, we propose a middleware, called the GX-plug, for the ease of integrating the merits of both.
As a middleware, the GX-plug is versatile in supporting different runtime environments, computation models, and programming models.
More, for improving the middleware performance, we study a series of techniques, including pipeline shuffle, synchronization caching and skipping, and workload balancing, for intra-, inter-, and beyond-iteration optimizations, respectively.
Experiments show that our middleware efficiently plugs accelerators to representative distributed graph systems, e.g., GraphX and Powergraph, with up-to 20x acceleration ratio.

\end{abstract}


%
%
%

\section{Introduction}

Big graph analytics are often with large data volumes, high computation intensiveness, and diversified applications, such as social networks, Internets, traffic networks, and biological structures, just to name a few.
To meet the scaling-out challenge \cite{HPCNeuroscience}, an increasing number of distributed graph systems, including GraphX \cite{DBLP:conf/osdi/GonzalezXDCFS14} and PowerGraph \cite{DBLP:conf/osdi/GonzalezLGBG12}, are proposed and deployed.
To meet the scaling-up challenge, non-distributed graph systems, such as Gunrock\cite{DBLP:conf/ppopp/WangDPWRO16} and ThunderGP\cite{DBLP:conf/fpga/ChenTCHWC21}, incorporate accelerators, including GPUs, multi-core CPUs, and FPGAs.
Beyond merely scaling-out or scaling-up, recent research spotlights the vision of a scalability continuum \cite{DBLP:journals/cacm/SakrBVIAAAABBDV21}, where
distributed graph systems and accelerators can be integrated for elastic scaling of big graph systems deployed in data centers.

The high computational concentration in cloud services makes an appealing case for accelerating applications on distributed systems. For example, multi-core processors, like GPUs~\cite{url:NVIDIASpark3} and multi/many-core CPUs~\cite{DBLP:conf/asplos/LindermanCWM08, DBLP:conf/ipps/HaidarCYLTKD14}, have been deployed as accelerators in distributed instances of cloud services, such as Amazon EC2, Google cloud,
Microsoft Azure Blob, and HW cloud, for flexibly scaling up the performance to application demands.
Also, Nvidia has announced the plan to support Spark $3.0$ with GPU acceleration in 2020~\cite{url:NVIDIASpark3, url:DatabrickSpark3}. So, it becomes a natural technology trend for integrating accelerators with distributed graph systems.

However, it is more challenging for distributed graph systems, because there exist a large number of system variants \cite{DBLP:journals/cacm/SakrBVIAAAABBDV21}, due to the diversity and irregularity of distributed graph processing.
They are with different architectures, runtime environments (Java and C++), 
and computation models\footnote{BSP (Bulk Synchronous Parallel) is a parallel model that performs computation in iterative steps, including three steps of computation, communication, and synchronization.
BSP model \cite{DBLP:conf/sigmod/MalewiczABDHLC10} has been the most fundamental and popular execution approach on distributed graph systems.
GAS (Gather-Apply-Scatter) model \cite{DBLP:conf/osdi/GonzalezLGBG12} is another basic and widely adopted model for distributed graph processing, based on BSP \cite{DBLP:journals/pvldb/AmmarO18}.}
and programming models\footnote{It includes vertex- and edge-centric models\cite{DBLP:journals/pvldb/Ozsu2019Keynote}.}.
In this work, we propose a middleware, the GX-Plug, where accelerators can be neatly plugged to heterogenous distributed graph systems.
With such a middleware, users can economically scale up their graph systems, to avoid the overhead of replanting to a new accelerator-aided graph systems, and save the efforts of accelerator accessing and subsequent optimizations.  

\begin{figure}[t]
    \centering
    \vspace{-5pt}
    \includegraphics[width=\linewidth]{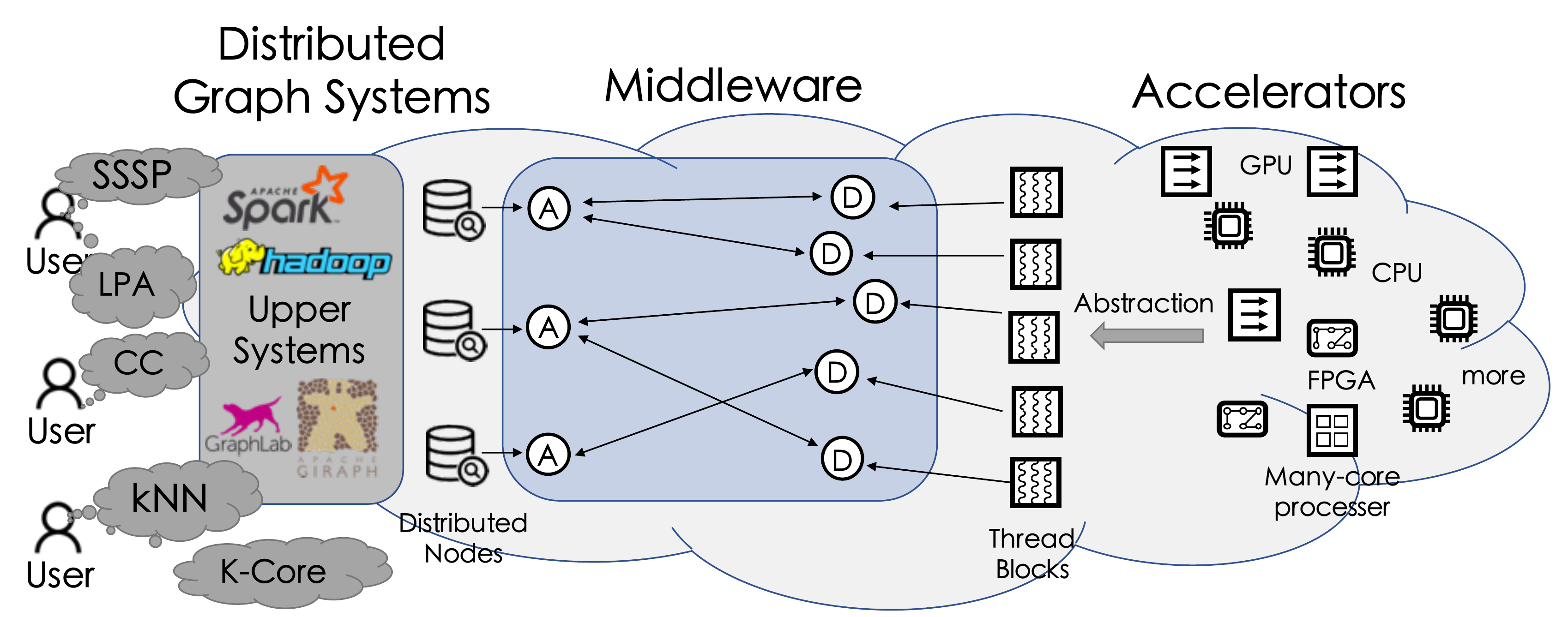}
    \vspace{-15pt}
    \caption{Middleware Overview (A: Agent, D: Daemon)}
    \label{fig:MOverview}
    \vspace{-5pt}
\end{figure}

A bird's eye view of the middleware is shown in Figure~\ref{fig:MOverview}, which follows an agile framework, called the {\it daemon-agent} framework. A daemon is a multi-core processor, an abstract representation of an accelerator, enabling transparent integration of accelerators (GPUs or CPUs), to a distributed graph system (GraphX or PowerGraph), by holding customizable graph programming interfaces.
An agent is resided in a distributed node for bridging upper systems and daemons, covering data exchanging and daemon life-cycle controlling.
A agent connects one or more daemons, according to the number of accelerators that the system allocates, for flexible computation resource distribution and workload balancing.
With the daemon-agent framework, the middleware shows flexibility in supporting different runtime environments, computation models, and programming models.

For easy accessing to accelerators, daemon provides a graph algorithm template, based on conventional iterative models, and support transplanting existing distributed graph algorithms with ease. Agent provides a kit of interfaces to cooperate with daemons and uppers systems for global graph computation.
Accordingly, it takes only a few lines of code to plug accelerators to upper systems.

%




Nevertheless, there arise a series of research challenges for implementing the middleware, besides the software design and development efforts on the adaption to different runtimes (Java or C++) and different accelerators (GPUs or CPUs).
First, there exists considerable data transmission overhead for the middleware in delivering and translating the data payloads into desired formats, between upper systems and accelerators, causing \emph{intra-iteration} overhead.
Second, the irregular and complex graph structure incurs imbalanced workloads, as well as latencies in frequent global synchronization, causing \emph{inter-iteration} overhead.
Third, it is difficult to schedule the workload and computation resource for different tasks and system configurations, recognized as {\it beyond-iteration} overhead.
The overheads can much degrade the middleware performance.

We tackle the first challenge by incorporating pipeline shuffle for optimizing the data transferring between daemons and agents.
We tackle the second challenge by optimizing the process of data synchronization, including caching and skipping, so as to minimize the volumes of data transferring during the synchronization phase.
We tackle the third challenge by making the size of transferred data blocks self-adaptive to the workloads of distribute nodes, and therefore the system workload balancing can be improved.

In this work, we focus on the implementation and optimization of the middleware. We are aware of optimization techniques, either on the accelerator end, e.g.,
exploring memory hierarchies \cite{DBLP:journals/corr/abs-1904-02241, DBLP:journals/taco/WangWLWZG21} for accelerating on-chip data accessing
or reinforcing local GPU processing networks with NVLink and NVSwitch\cite{8763922,url:nvlink};
or on the upper system end, e.g., using RDMA\cite{DBLP:conf/ccgrid/FuVSIY18} for faster distributed system communication and using pull-push model\cite{DBLP:journals/pvldb/JiaKSMEA17} for data transferring optimization in specific applications.
We would like to argue that optimizations merely on upper system or accelerator end are beyond the scope of the middleware, and are orthogonal to our work.

Our contributions can be listed as follows.


\begin{itemize}
    \item We propose, to our best knowledge, the first middleware for arming distributed graph systems with high-performance accelerators, to meet the needs of scaling-out and -up in big graph analytics.
    \item For the middleware, we design a novel daemon-agent framework, which achieves flexible deployment on different upper systems and easy accessing to accelerators.
    \item The middleware is general in supporting different computation models, such as BSP and GAS. Existing distributed graph algorithms can be transplanted for accessing accelerators with ease.
    \item For the middleware optimization, we investigate a series of techniques, such as pipeline shuffling, synchronization caching and skipping, and workload balancing, for intra-, inter-, and beyond-iteration optimizations, respectively.
    \item We conduct extensive experiments on real datasets to evaluate the efficiency and scalability of the middleware.
\end{itemize}

The rest of the paper is organized as follows.
Section~\ref{sec:mo} shows the overview of the middleware.
Section~\ref{sec:ro} investigates optimization techniques used to improve the internal performance of middleware.
Section~\ref{sec:si} discusses middleware deployment techniques.
Section~\ref{sec:ev} reports the results of empirical studies.
Section~\ref{sec:relate} presents related works.
Section~\ref{sec:con} concludes the paper.
Our middleware is open sourced\footnote{https://thoh-testarossa.github.io/GX-Plug/\label{source}}.

\section{Middleware Overview}
\label{sec:mo}

%
%
%

Cloud services are witnessed to evolve from cloud storage services, comprised of a multitude of distributed nodes/machines/instances, to high-performance cloud computing services, comprised of accelerator-powered distributed nodes, as aforementioned. 
Our middleware is to boost graph computing on such cloud services, supporting system configuration and application development with ease.
In this section, we investigate the daemon-agent framework, which is the core of the middleware, in Section~\ref{subsec:framework}. Then, we study the data storage and the controllers of the middleware, in Sections~\ref{subsec:DataManagement} and \ref{subsec:ctrl}, respectively.




\subsection{Daemon-Agent Framework}
\label{subsec:framework}

The structures of daemons and agents, and their interactions are shown in Figure~\ref{fig:DAFOverview}. In general, daemons are in connection with accelerators, and agents are in connection with upper systems.
The communication of the two parts is done via the System V IPC.

\begin{figure}[ht]
    \centering
    \vspace{-5pt}
    \includegraphics[width=\linewidth]{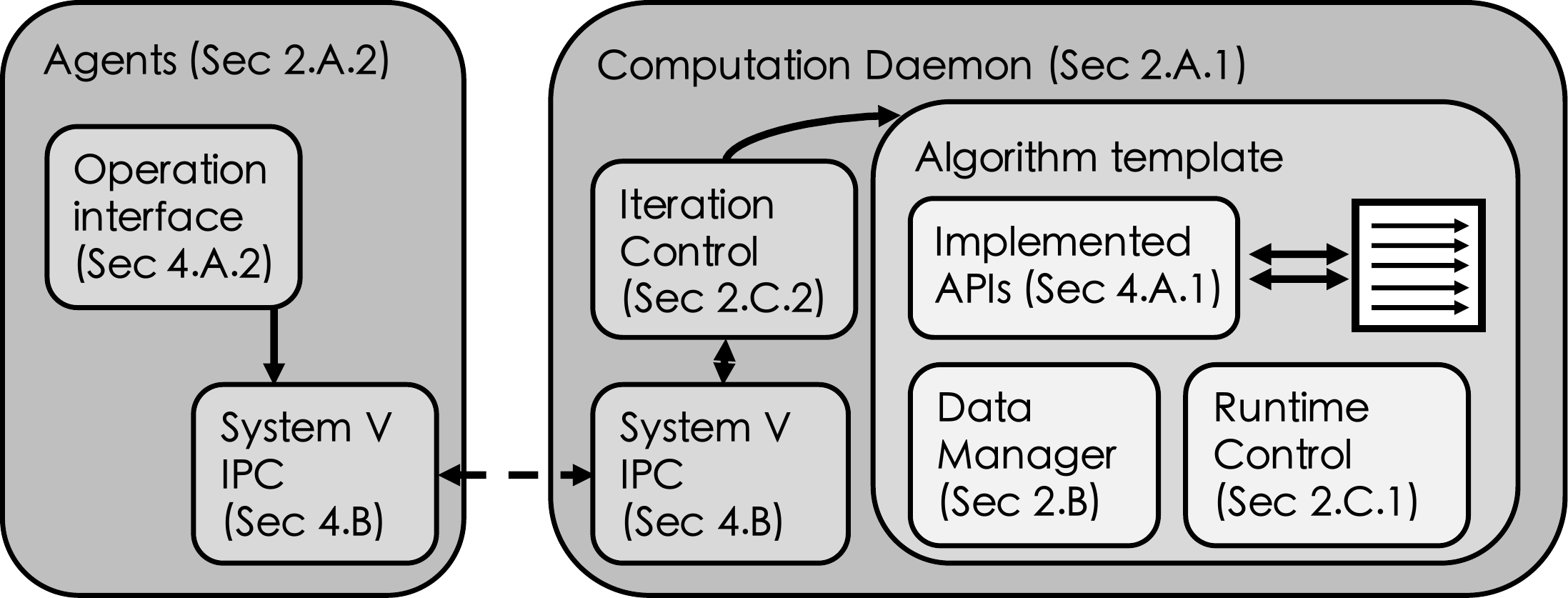}
    \vspace{-10pt}
    \caption{Daemon-Agent Framework}
    \label{fig:DAFOverview} 
    \vspace{-5pt}
\end{figure}



\subsubsection{Daemon}

\label{subsubsec:daemon}

A daemon represent an accelerator, where graph algorithms are executed.
A daemon thus holds an algorithm template and the iteration logic controlling, as shown in Figure~\ref{fig:DAFOverview}.
The design of a daemon is towards transparent hardware management for upper systems.
At the runtime, an instance of the algorithm template is implemented for daemons.
For accelerating distributed graph algorithms,
algorithm engineers only focus on the implementation of the APIs of the algorithm template.
The connection with accelerators are established during the initialization phase, and details are hidden to system developers after that.



\subsubsection{Agent}

\label{subsubsec:agent}

An agent represent a distributed node of an upper system and makes a bridge for upper systems and daemons.
Essentially, an agent covers a set of operation interfaces between upper systems and daemons, on data exchanging, subfunction execution, and daemon lifecycle controlling.
With the operation interfaces on agents, upper systems can substantially configure and control daemons, including specifying the number of accelerators and mixing and matching different types of accelerators in a system.
The structure of an agent is shown in Figure \ref{fig:DAFOverview}.


In the local environment of a distributed node, there should be at least one agent, 
and one or more daemons, representing different accelerators.
Also, as shown in Section~\ref{sec:ro}, agents are equipped with a series of optimization techniques to reduce the overhead caused by data transferring, which is the major source affecting system performances.

%
%
%

\subsection{Data Flows \& Management}
\label{subsec:DataManagement}


%
%
%
%
%


The challenges in the data management of the middleware are two-fold:
1) upper systems and accelerators can be of different runtime environments (C++ and Java);
2) data in different upper systems may follow different (vertex- or edge-centric) storage strategies. More, the data transferring should be efficient in order to meet the system runtime requirements.

The data flow in the middleware is shown in Figure~\ref{fig:DataStorageOverview}.
To tackle the efficiency challenge, the graph data is neither stored in the agent side, nor in the daemon side. Instead, data is stored in the shared memory space based on the System V IPC.


{
\begin{figure}[ht]
    \centering
    \vspace{-5pt}
    \includegraphics[width=0.8\linewidth]{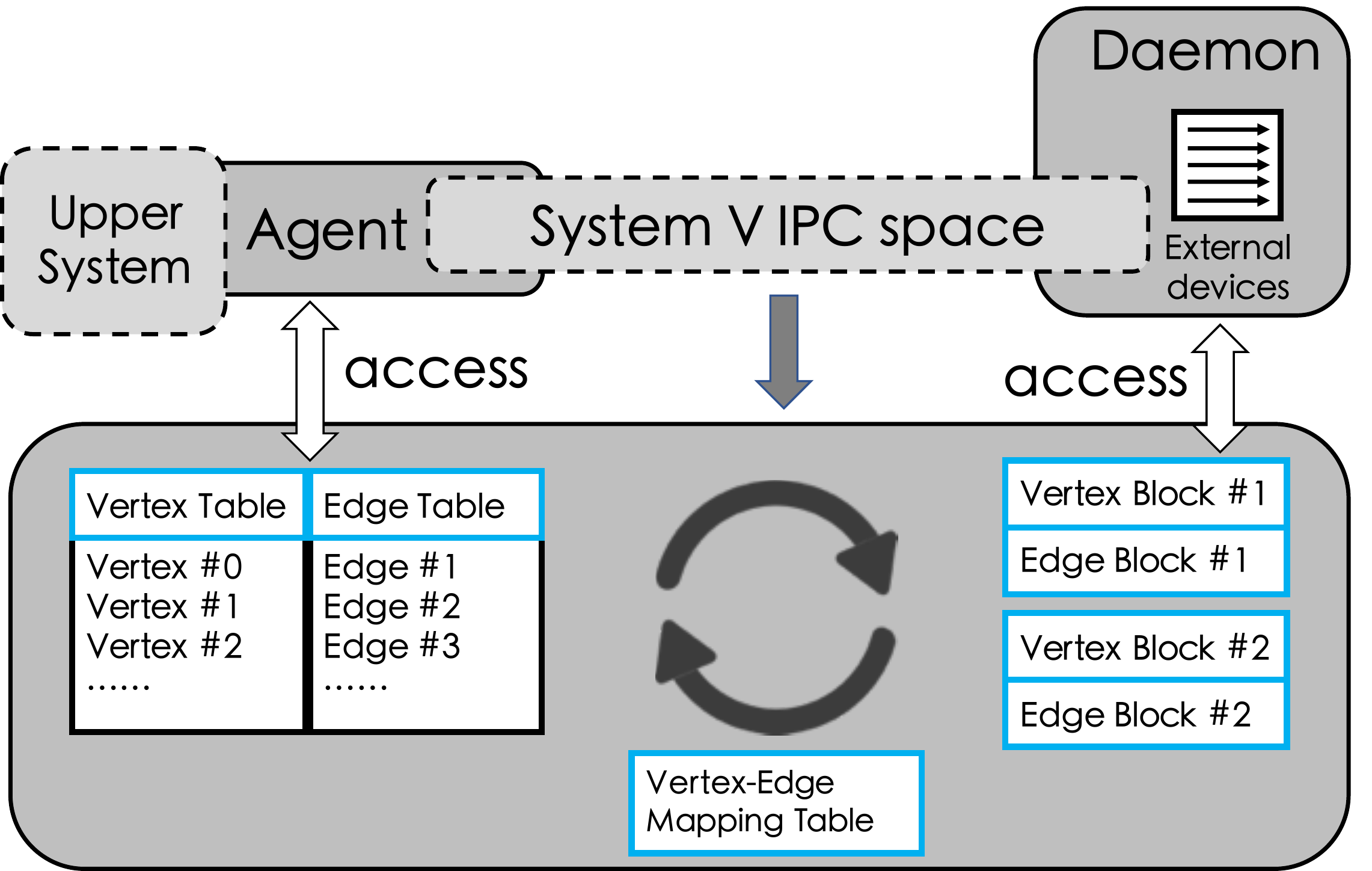}
    \caption{Data Flows \& Management}
    \label{fig:DataStorageOverview}
    \vspace{-5pt}
\end{figure}
}


Initially, the graph data are partitioned to distributed nodes by upper systems.
Then, for each distributed node, the data are fed to daemons for acceleration via agents.
However, data accessed by an agent cannot be directly accessed by a daemon, since they belong to different processes with no common memory space, as discussed in Section~\ref{subsec:RuntimeIsolation}. More, conventional inter-process communication incurs extra data transferring, degrading the system runtime performance.
In our implementation, we use kernel functions aided by the UNIX System V to create a shared memory space for daemons and agents. 
In our middleware, a daemon has a unique System V key pointing to its specific shared memory space, while an agent has multiple keys to communicate with all daemons attached to it.

The benefits are on three aspects.
First, the common shared memory space enables the mutual data accessing between upper systems and accelerators that are of different environments.
Second, data accessing can be done via the common memory space, avoiding the intermediate data copying between the two ends.
Third, any data updates in the agent or daemon end can be immediately perceived by the other end without extra sensing efforts, thus facilitating the control logic of the system.





Based on the common memory space, we proceed to discuss the data management in the middleware.
A agent uses a vertex table and an edge table to manage the graph data of a distributed node. The structure is also general in supporting different (vertex- or edge-centric) storage strategies in upper systems.
For daemons, the edge-centric strategy is adopted, because it is commonly accepted to be more effective in workload balancing, than its counterpart \cite{nature, DBLP:conf/osdi/GonzalezLGBG12}, especially for real graphs following power-law distributions \cite{donato2004large, DBLP:conf/icde/KDY21}.
For efficient processing in accelerators, a daemon uses a series of data blocks, including vertex blocks and edge blocks, to be fed to accelerators.
Each edge block contains a fixed number of edges.
Also, each edge block is associated with a paired vertex block, where both source and destination vertices of an edge can be found.
There is a vertex-edge mapping table, for transforming the data stored in the vertex and edge tables of the agent end to the vertex and edge blocks of the daemon end.
Thus, to construct an edge block, an agent selects a vertex and retrieves its outer edges, with vertex-edge mapping table.
The corresponding vertex block is constituted by incorporating destination vertices, as well as their attributes, for the edges in the edge block.

So, at each iteration of computation, the middleware packages up the vertex and edge blocks for accelerators, by repeatedly selecting vertices or edges that are needed.
After an iteration of computation, updated vertex and edge blocks are synchronized back to the vertex and edge tables.

\subsection{Controllers}
\label{subsec:ctrl}

For the middleware, daemons are in charge of orchestrating different computational components. We introduce two components located in the daemon, making the middleware adaptable to different computation models and optimizing the iterative processing for upper systems.

\subsubsection{Runtime Control}
\label{subsubsec:runctrl}

The runtime control component is for controlling the execution order of implemented APIs, including the runtime information collected from accelerators and sending/receiving flags for iteration controlling.
By controlling the execution order of implemented APIs, the middleware can easily be integrated into different computation models, as discussed in Section~\ref{subsubsec:powergraph}.

\subsubsection{Iteration Control}
\label{subsubsec:iterctrl}

The iteration control component is for controlling and coordinating the entire iteration. Since the middleware separates the runtimes of upper systems and the computation, it is necessary to connect both parties for in-between data synchronization and the computation processing cycle.
Several optimizations, such as pipeline shuffling, synchronization caching, and skipping, are implemented in Section~\ref{sec:ro}.
The component collaboratively works with agents for retrieving information from upper systems, and works with the runtime control component for exchanging flag information, to fulfill the iteration controlling.

The main goal of the middleware is to shield the system developers and graph algorithm programmers from the heterogeneity in different systems and accelerators.
It is achieved at the expense of internal overheads in the middleware.
In the sequel, we devise a series of optimizations to alleviate or even eliminate the overhead originating from the middleware.  

\section{Runtime Optimization}

\label{sec:ro}

In this section, we introduce three optimization techniques, pipeline shuffling for improving intra-iteration processing, synchronization caching and skipping for eliminating unnecessary data transferring for inter-iteration processing, and workload balancing for beyond-iteration optimization.

\subsection{Intra-Iteration Optimization: Pipeline Shuffling}

\label{subsec:pipelineShuffle}

\subsubsection{Motivation}

For the basic daemon-agent framework aforementioned, an ordinary workflow of graph processing acceleration consists of five steps, data downloading (from upper systems), agent-to-daemon data transferring, computing, daemon-to-agent data transferring, and data uploading (to upper systems).
However, a tightly coupled execution of the five steps, where the output of one step is streamed as the input of another step, leads to many waiting-and-suspending states and therefore the underutilization of computation resources.

For example, the computing step must wait for agent-to-daemon data transferring to start, so that the computing step would be suspended during other steps.
To alleviate the predicament and to improve the computation resource utilization, we investigate a pipeline parallelism mechanism, \textit{Pipeline shuffle}, to the middleware. 

\subsubsection{Overview}




\begin{figure}[ht]
    \centering
    \vspace{-5pt}
    \includegraphics[width=\linewidth]{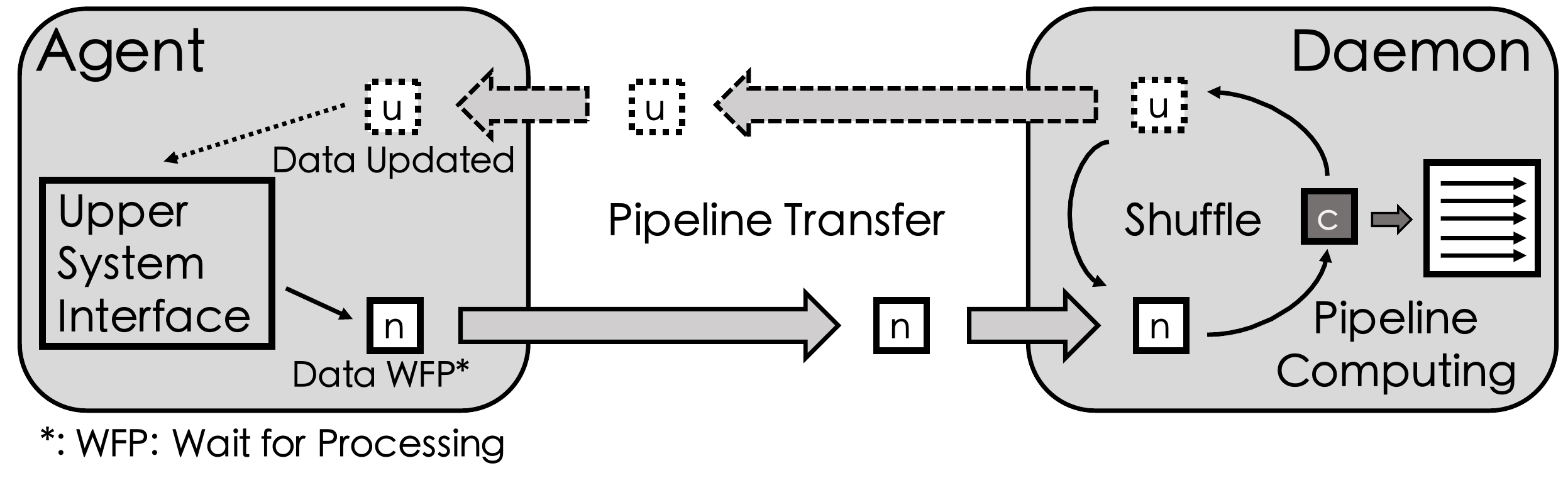}
    \vspace{-15pt}
    \caption{Pipeline Shuffle}
    \label{fig:PipelineShuffleOverview}
    \vspace{-10pt}
\end{figure}

\begin{algorithm}[ht]
    \small
    \caption{Pipeline Shuffle - Daemon Side}
    \label{al:PipelineShuffleDaemon}
    \algorithmicrequire{Computer Device $com\_dev$, Data area pointer $n$, $c$, $u$}
    \begin{algorithmic}[1]
        \While{\textnormal{In Iteration}}
            \State Block\_Recv($agent$, $msg$)\;
            \If{$msg$ = \textnormal{``ExchangeFinished''}}
                \State Rotate($n$ $\rightarrow$ $c$ $\rightarrow$ $u$ $\rightarrow$ $n$)\;
                \State Send($agent$, ``RotateFinished'')\;
            \ElsIf{$c$ \textnormal{contains contents to compute}}
                \State $com\_dev$.Load($*c$)\;
                \State $com\_dev$.Compute()\;
                \State $*c$ $\leftarrow$ $com\_dev$.data\;
                \State Send($agent$, ``ComputeFinished'')\;
            \Else
                \State Send($agent$, ``ComputeAllFinished'')\;
                \State End\_Iteration()\;
            \EndIf
        \EndWhile
    \end{algorithmic}
\end{algorithm}

\begin{algorithm}[ht]
    \small
    \caption{Pipeline Shuffle - Agent Side}
    \label{al:PipelineShuffleAgent}
    \algorithmicrequire{Upper system interface $USI$, Data area pointer $n$, $c$, $u$}
    \begin{algorithmic}[1]
        \State $*n$ $\leftarrow$ $USI$.Download()\;
        \State Send($daemon$, ``ExchangeFinished'')\;
        \While{\textnormal{In Iteration}}
            \State Block\_Recv($daemon$, $msg$)\;
            \If{$msg$ = \textnormal{``RotateFinished''}}
                \State $upload$, $download$ = new Thread()\;
                \For{\textnormal{Thread} upload}
                    \State $USI$.Upload($*u$)\;
                \EndFor
                \For{\textnormal{Thread} download}
                    \State $*n$ $\leftarrow$ $USI$.Download()\;
                \EndFor
            \ElsIf{$msg$ = \textnormal{``ComputeFinished''}}
                \If{$upload$\textnormal{.isTerminated()}}
                    \If{$download$\textnormal{.isTerminated()}}
                        \State Send($daemon$, ``ExchangeFinished'')\;
                    \EndIf
                \EndIf
            \ElsIf{$msg$ = \textnormal{``ComputeAllFinished''}}
                \If{$upload$\textnormal{.isTerminated()}}
                    \If{$download$\textnormal{.isTerminated()}}
                        \State End\_Iteration()\;
                    \EndIf
                \EndIf
            \EndIf
        \EndWhile
    \end{algorithmic}
\end{algorithm}

The idea is to construct a multi-layer pipeline for better parallelism.
First, we replace the original $5$-step data transferring into a $3$-step data transferring, data downloading, computing, and data uploading, which are handled by $3$ threads, Thread.Download, Thread.Compute and Thread.Upload, respectively.
Compared to the $5$-step setting, the $3$-step setting eliminates the two steps of agent-to-daemon and daemon-to-agent data transferring.
Second, based on the $3$-step setting, we construct a $3$-layered pipeline to reduce the suspension time of the computing step, and finally improve the accelerator utilization.
The overview of pipeline shuffle is shown in Figure~\ref{fig:PipelineShuffleOverview}.
The detailed process is depicted by Algorithms~\ref{al:PipelineShuffleDaemon} and \ref{al:PipelineShuffleAgent}.

\paragraph{Pipeline Parallelism}




The pipeline consists of $3$ layers in correspondence to the $3$ steps mentioned above, as shown in Figure~\ref{fig:PipelinedProcess}.
With the pipeline, an iteration can be decomposed into a sequence of $3$ processing cycles, corresponding to the $3$ layers.
For all three pipeline layers,
we use ``edge triplets'' as the intermediate data structure, which includes an edge and its source and destination vertices, by efficiently joining the edge and vertex tables.
With the data structure of edge triplets, the pipeline has homogenous data structures for all layers, for avoiding unnecessary data format transformation and enabling grained-granularity data retrieval.
Essentially, the triplet is the basic processing unit of an iteration,
which serves as both the source of computation input and the carrier of output.
Within an iteration, there is no data dependencies between triplets.

For each layer, triplets are grouped into a set of blocks, as shown in Figure~\ref{fig:PipelinedProcess}. The blocks are assigned to the $3$ threads for processing. Thus, pipeline parallelism can be established which significantly improves the system performance.

\begin{figure}[ht]
    \centering
    \vspace{-5pt}
    \includegraphics[width=\linewidth]{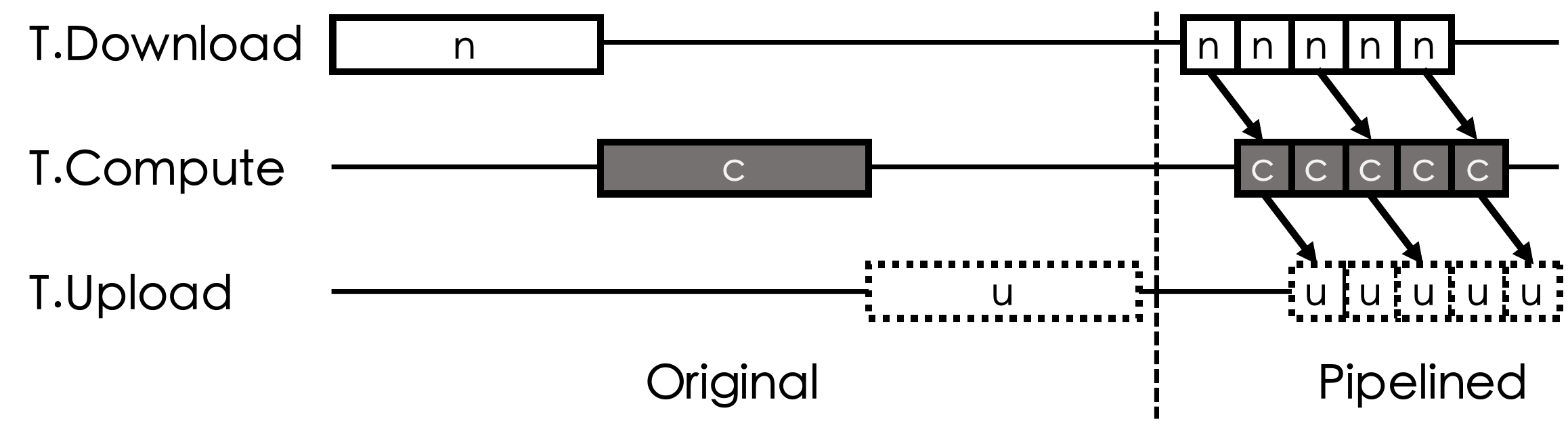}
    \vspace{-15pt}
    \caption{Pipelined Process Flow with Pipeline Shuffle}
    \label{fig:PipelinedProcess}
    \vspace{-5pt}
\end{figure}



\paragraph{Shuffle for data transfer reduction}

There are intermediate data layers existing in an ordinary pipeline.
The $i$-th data layer stores the result of $i$-th pipeline process cycle, and is used for $(i+1)$-th downstream pipeline process cycle.
Within a cycle, data are transferred between threads for fulfilling the processing.
However, frequent data copying incurs considerable system overhead.

To handle that, we design a shuffle mechanism for efficient data transferring in pipeline parallelism. First, we allocate equal sized memory chunks for threads. Each memory chunk is associated with a pointer, for the reference of the front block currently being processed inside the chunk.
Second, inter-thread data copying is replaced by the pointer copying. For example, in Figure~\ref{fig:PipelinedProcess}, data blocks of $3$ layers can be represented by $n$-block, $c$-block, and $u$-block, indicating blocks for new data retrieved from upper systems, blocks for computing, and blocks for uploading to upper systems, respectively. When a pipeline cycle is finished, the $3$ pointers are shuffled in a rotation manner: the pointer to $n$-block switches to $c$-block, the pointer to $c$-block switches to $u$-block, and so on.
With the pipeline shuffle mechanism, there is no more data copying between threads, since it is completed in situ.


\paragraph{Block Size Selection}

In our work, we found the size of a block has a profound effect over the parallelism performance.
We assume that there is a sub-dataset distributed to a agent-daemon pair for processing which contains $d$ entities need to be processed in the current iteration. Also, agent divides the dataset into $s$ blocks evenly, $b = \frac{d}{s}$. Let $T_{n}(b)$, $T_{c}(b)$, $T_{u}(b)$ be the process time of one data block in Thread.Download, Thread.Compute and Thread.Upload, respectively.
$T_n$ and $T_u$ are proportional to the data block size.
We can estimate pipeline processing time $T_{total}$.
\begin{equation}
 \small   \begin{aligned}
        T_{total} = &
        T_{n}(b) + \max(T_{n}(b), T_{c}(b))\\
        + & (s - 2)\max(T_{n}(b), T_{c}(b), T_{u}(b))\\
        + & \max(T_{c}(b), T_{u}(b)) + T_{u}(b)
    \end{aligned}
    \label{eql:PipelineTimeSimplied}
\end{equation}

$T_{c}$ refers to the time cost of Thread.Compute, and consists of calling computation devices, loading data to devices, and computation. Their corresponding time costs are represented by $T_{call}$, $T_{comp}$, and $T_{copy}$, respectively.
The operation of calling devices takes constant time, while computation and data copying time are related to data size. So, $T_c(b)$ can be modeled as follows.
\begin{equation}
\small
    \begin{aligned}
        \nonumber T_{c}(b) = T_{call} + T_{comp}(b) + T_{copy}(b)
    \end{aligned}
    \label{eql:PipelineTimeComp}
\end{equation}

When $s$ increases, block size $b$ decreases, so do $T_{n}$ and $T_{u}$. $T_{c}$ also decreases when $b$ decreases, but will never be less than $T_{call}$, which means $T_{total}$ starts increasing when $s$ is large enough and keep increasing. On the other hand, $T_{n}$ and $T_{u}$ increase when $s$ is being smaller, since $b$ is being larger. Thus, both the function $T_{total}(b)$ and $T_{total}(s)$ should tend to become a U-turn form.
Thus, $s$ and $b$ should be deliberately configured for achieving fine-tuned system performance. We will try to calculate the value of $b$ in order to provide optimization suggestion to overall system.

The calculation follows two assumptions: 1) The sub-dataset distributed to the agent-daemon pair has $d$ entities, and is divided into $s$ blocks evenly, which have the size $b = \frac{d}{s}$. 2) we assume both $T_{n}$, $T_{u}$ are directly proportional to the block size; and 3) In $T_{c}$ shown in Equation~\ref{eql:PipelineTimeComp}, $T_{call}$ has a fixed number, while $T_{comp}$ and $T_{copy}$ are also directly proportional to $b$. We use $k_{1}$, $k_{2}$ and $k_{3}$ to represent download, computation, and upload time cost per unit of data entity. Thus, Equation~\ref{eql:PipelineTimeSimplied} can be further simplified as follows:
\begin{equation}
\small
    \begin{aligned}
        T_{total} =
        k_{1}b + \max(k_{1}b, a + k_{2}b)\\
        + (s - 2)\max(k_{1}b, a + k_{2}b, k_{3}b)\\
        + \max(a + k_{2}b, k_{3}b) + k_{3}b
    \end{aligned}
    \label{eql:PipelineTimeSimplied2}
\end{equation}

Lemma \ref{lemma:bSizeSelection} shows the derivation of optimal $b$ for the pipeline shuffle mechanism, as follows.

\begin{lemma}
\label{lemma:bSizeSelection}
        If a distributed node stores $d$ data entities, the optimal block size $b_{opt}$ and corresponding $T_{total_{min}}$ can be calculated as follows, where $Q = \sqrt{\frac{ad}{k_{1} + k_{3}}}$.
        \begin{equation}
        \small \nonumber
            b_{opt} = \left\{
            \begin{aligned}
                \frac{a}{k_{1} - k_{2}} & , & \left(
                    \begin{aligned}
                        k_{max} = k_{1} & ,\& \\
                        \frac{a}{k_{1} - k_{2}} < Q & \\
                    \end{aligned}
                \right)\\
                \frac{a}{k_{3} - k_{2}} & , & \left(
                    \begin{aligned}
                        k_{max} = k_{3} & ,\& \\
                        \frac{a}{k_{3} - k_{2}} < Q & \\
                    \end{aligned}
                \right)\\
                \sqrt{\frac{ad}{k_{1} + k_{3}}} & , & otherwise\\
            \end{aligned}
            \right. \text{~~~~~~, and}
            \label{eql:PipelinebFinal}
        \end{equation}
        \begin{equation}
        \small    T_{total_{min}} = \left\{
            \begin{aligned}
                \frac{a(k_{1} + k_{3})}{k_{1} - k_{2}} + k_{1}d & , & \left(
                    \begin{aligned}
                        k_{max} = k_{1} & ,\& \\
                        \frac{a}{k_{1} - k_{2}} < Q & \\
                    \end{aligned}
                \right)\\
                \frac{a(k_{1} + k_{3})}{k_{3} - k_{2}} + k_{3}d & , & \left(
                    \begin{aligned}
                        k_{max} = k_{3} & ,\& \\
                        \frac{a}{k_{3} - k_{2}} < Q & \\
                    \end{aligned}
                \right)\\
                k_{2}d + 2\sqrt{(k_{1} + k_{3})ad} & , & otherwise\\
            \end{aligned}
            \right.
            \label{eql:PipelineTimeFinal}
        \end{equation}
        
\end{lemma}

\begin{proof}
    Following Equation~\ref{eql:PipelineTimeSimplied2}, we have $3$ cases to consider.

    \underline{\textbf{Case 1: $T_{n} = k_{1}b$ is the maximum value.}}
    This case is true only if 
    $k_{1}$ is the maximum of the $3$ parameters, $k_1$, $k_2$, and $k_3$.
    Accordingly, $b$ should satisfy the follows.

    \begin{equation}
    \small    \begin{aligned}
            \nonumber k_{1}b \geq a + k_{2}b \Rightarrow
            \nonumber b \geq \frac{a}{k_{1} - k_{2}}
        \end{aligned}
        \label{eql:PipelineTimek1Max}
    \end{equation}

    Thus, Equation~\ref{eql:PipelineTimeSimplied2} can be transformed into:

    \begin{equation}
    \small    \begin{aligned}
            \nonumber T_{total} =
            \nonumber sk_{1}b + \max(a + k_{2}b, k_{3}b) + k_{3}b\\
            \nonumber = k_{1}d + \max(a + k_{2}b, k_{3}b) + k_{3}b\\
        \end{aligned}
    \end{equation}

    Notice that $a$ and $\{k_{i}\}_{i \leq 3}$ are all positive, and both $\max(a + k_{2}b, k_{3}b)$ and $k_{3}b$ increase when $b$ increases. Thus, when $b = \frac{a}{k_{1} - k_{2}}$, we have the minimum value of $T_{total}$ as follows.

    \begin{equation}
     \small   \begin{aligned}
            \nonumber
            T_{total} = k_{1}d + \max(a + \frac{ak_{2}}{k_{1} - k_{2}}, \frac{ak_{3}}{k_{1} - k_{2}}) + \frac{ak_{3}}{k_{1} - k_{2}}\\
            = k_{1}d + \max(\frac{ak_{1}}{k_{1} - k_{2}}, \frac{ak_{3}}{k_{1} - k_{2}}) + \frac{ak_{3}}{k_{1} - k_{2}} \\
            = k_{1}d + \frac{(k_{1} + k_{3})a}{k_{1} - k_{2}}\\
        \end{aligned}
        \label{eql:PipelineTimek1MaxFinal}
    \end{equation}

    \underline{\textbf{Case 2: $T_{c} = (a + k_{2}b)$ is the maximum value.}}
    First, we have this equation below, where $s = \frac{d}{b}$.
    \begin{equation}
    \small    \begin{aligned}
            \nonumber
            T_{total} =
            k_{1}b + s(a + k_{2}b) + k_{3}b\\
            = (k_{1} + k_{3})b + k_{2}d + \frac{ad}{b}\\
        \end{aligned}
        \label{eql:PipelineTimek2MaxM}
    \end{equation}

    In this equation, we can have the minimum $T_{total}$, if $b$ equals $\sqrt{\frac{ad}{k_{1} + k_{3}}}$, which is $Q$.
    Notice that $b$ may not equal $Q$, as constrained by $\{k_i\}$.
    Accordingly, we discuss $T_{total}$ in $3$ subcases, based on the value of $k_{2}$.

    \underline{\textit{($k_{2}$ is the minimum one. )}}
    In this situation, both $(k_{1} - k_{2})$ and $(k_{3} - k_{2})$ are positive. Thus, we have:
    \begin{equation}
    \small
        \begin{aligned}
            \nonumber a + k_{2}b \geq \max(k_{1}, k_{3}) \cdot b
            \Rightarrow
            \nonumber b \leq \min(\frac{a}{k_{1} - k_{2}}, \frac{a}{k_{3} - k_{2}})
        \end{aligned}
        \label{eql:PipelineTimek2MaxS1}
    \end{equation}

    Assume that $k_{1} \geq k_{3}$, we have $b \leq \frac{a}{k_{1} - k_{2}}$. Thus, we have the minimum $T_{total}$:
    \begin{equation}
    \small
        \nonumber T_{total} = \left\{
        \begin{aligned}
            k_{2}d + 2\sqrt{(k_{1} + k_{3})ad} & , & \frac{a}{k_{1} - k_{2}} \geq Q\\
            \frac{a(k_{1} + k_{3})}{k_{1} - k_{2}} + k_{2}d + (k_{1} - k_{2})d& , & \frac{a}{k_{1} - k_{2}} < Q\\
        \end{aligned}
        \right.
        \label{eql:PipelineTimek2MaxS1Final}
    \end{equation}
    
    Also, we have the minimum $T_{total}$ when $k_{3} \geq k_{1}$:
    \begin{equation}
    \small
        \nonumber T_{total} = \left\{
        \begin{aligned}
            k_{2}d + 2\sqrt{(k_{1} + k_{3})ad} & , & \frac{a}{k_{3} - k_{2}} \geq Q\\
            \frac{a(k_{1} + k_{3})}{k_{3} - k_{2}} + k_{2}d + (k_{3} - k_{2})d& , & \frac{a}{k_{3} - k_{2}} < Q\\
        \end{aligned}
        \right.
        \label{eql:PipelineTimek2MaxS1Final02}
    \end{equation}

    \underline{\textit{($k_{2}$ is the middle one. )}}
    In this situation, we should notice the change of the inequality, because some terms of Equation~\ref{eql:PipelineTimek2MaxS2} can be negative.
    Without loosing generality, we assume $k_{3} \leq k_{2} \leq k_{1}$. In this case, $(k_{1} - k_{2})$ is positive, and $(k_{3} - k_{2})$ is negative. Thus, we can have $b \leq \frac{a}{k_{1} - k_{2}}$, since
    \begin{equation}
     \small
        \begin{aligned}
            \frac{a}{k_{3} - k_{2}} < 0 \leq b \leq \frac{a}{k_{1} - k_{2}} \Rightarrow 0 \leq b \leq \frac{a}{k_{1} - k_{2}}
        \end{aligned}
        \label{eql:PipelineTimek2MaxS2}
    \end{equation}
    
    Then, we have the minimum value of $T_{total}$ as shown in Equation~\ref{eql:PipelineTimek2MaxM}.
    \begin{equation}
    \small
        T_{total} = \left\{
        \begin{aligned}
            \nonumber
            \frac{a(k_{1} + k_{3})}{k_{1} - k_{2}} + k_{2}d + (k_{1} - k_{2})d& , & \frac{a}{k_{1} - k_{2}} < Q\\
            k_{2}d + 2\sqrt{(k_{1} + k_{3})ad} & , & otherwise\\
        \end{aligned}
        \right.
        \label{eql:PipelineTimek2MaxS2Final}
    \end{equation}
    
    On the other hand, if $k_{3} \geq k_{2} \geq k_{1}$ holds, the minimum value of $T_{total}$ is:
    \begin{equation}
    \small 
        T_{total} = \left\{
        \begin{aligned}
            \nonumber
            \frac{a(k_{1} + k_{3})}{k_{3} - k_{2}} + k_{2}d + (k_{3} - k_{2})d& , & \frac{a}{k_{3} - k_{2}} < Q\\
            k_{2}d + 2\sqrt{(k_{1} + k_{3})ad} & , & otherwise\\
        \end{aligned}
        \right.
        \label{eql:PipelineTimek2MaxS2Final02}
    \end{equation}
    
    \underline{\textit{($k_{2}$ is the maximum one. )}}
    In this situation, both $k_{1} - k_{2}$ and $k_{3} - k_{2}$ are negative. Since $b > 0$, $b$ is also greater than $\frac{a}{k_{1} - k_{2}}$ and $\frac{a}{k_{3} - k_{2}}$. Thus, we have the minimum $T_{total}= k_{2}d + 2\sqrt{(k_{1} + k_{3})ad}$, when $b = \sqrt{\frac{ad}{k_{1} + k_{3}}} = Q$.
    
    \underline{\textbf{Case 3: $T_{u} = k_{3}b$ is the maximum value.}}
    This case is true only if $k_3$ is the maximum of the $3$ parameters in $k_1$, $k_2$, and $k_3$.

    Following the discussion in Equation~\ref{eql:PipelineTimek1MaxFinal}, we simply have the conclusion that $T_{total}$ has the minimum value $k_{3}d + \frac{(k_{1} + k_{3})a}{k_{3} - k_{2}}$ when $b = \frac{a}{k_{3} - k_{2}}$.

    \textbf{Discussion.} Following the previous discussion, and the order of $k_{1}$, $k_{2}$ and $k_{3}$, we have $3$ cases to calculate $b_{opt}$.

    Case (i). $k_{1}$ is the maximum one: if $\frac{a}{k_{1} - k_{2}} \geq \sqrt{\frac{ad}{k_{1} + k_{3}}} = Q$, $b = \sqrt{\frac{ad}{k_{1} + k_{3}}} = Q$, and $T_{total}$ have the minimum value $k_{2}d + 2\sqrt{(k_{1} + k_{3})ad}$. Otherwise, $b = \frac{a}{k_{1} - k_{2}}$, and $T_{total}$ have the minimum value $k_{1}d + \frac{(k_{1} + k_{3})a}{k_{1} - k_{2}}$.

    Case (ii). $k_{2}$ is the maximum one: $b = \sqrt{\frac{ad}{k_{1} + k_{3}}} = Q$, and $T_{total}$ have the minimum value $k_{2}d + 2\sqrt{(k_{1} + k_{3})ad}$.

    Case (iii). $k_{3}$ is the maximum one: if $\frac{a}{k_{3} - k_{2}} \geq \sqrt{\frac{ad}{k_{1} + k_{3}}} = Q$, $b = \sqrt{\frac{ad}{k_{1} + k_{3}}} = Q$, and $T_{total}$ have the minimum value $k_{2}d + 2\sqrt{(k_{1} + k_{3})ad}$. Otherwise, $b = \frac{a}{k_{3} - k_{2}}$, and $T_{total}$ have the minimum value $k_{1}d + \frac{(k_{1} + k_{3})a}{k_{3} - k_{2}}$.

    Then, Equations in Lemma~\ref{lemma:bSizeSelection} can be directly derived from the above discussion, by grouping $b$ and $k$.
\end{proof}

Notice that both $s$ and $b$ must be integers. If $b_{opt}$ or $s_{opt} = \frac{d}{b_{opt}}$ is not an integer, we choose $2$ values $\lfloor s_{opt} \rfloor$ and $\lceil s_{opt}\rceil$ for $s$, and $2$ values $\lfloor b_{opt} \rfloor$ and $\lceil b_{opt} \rceil$ for $b$, so that Equation~\ref{eql:PipelineTimeSimplied2} can be used for estimating the minimum $T_{total}$.


\subsection{Inter-Iteration Optimization: Synchronization Caching and Skipping}

\label{subsec:synccaching}

\subsubsection{Motivation}


For a distributed graph system,
there is inevitable data synchronization between iterations,
for ensuring the data correctness in every distributed node.
However, it is often costly to do such synchronizations, since it would trigger considerable data copying between two successive iterations.
Also, in a na\"{i}vely integrated scale-out and -up system, the data copying involves memory accessing from distributed system environments to external computation accelerators, which is costly.
Thus, it is necessary to reduce the load of data synchronization, either for the number of times that synchronization triggers, or for the data volume transferred.
To this end, we introduce techniques of {\it synchronization caching} and {\it synchronization skipping} for inter-iteration optimization.

\subsubsection{Synchronization Caching}

\begin{figure}[ht]
    \centering
    \vspace{-5pt}
    \includegraphics[width=\linewidth]{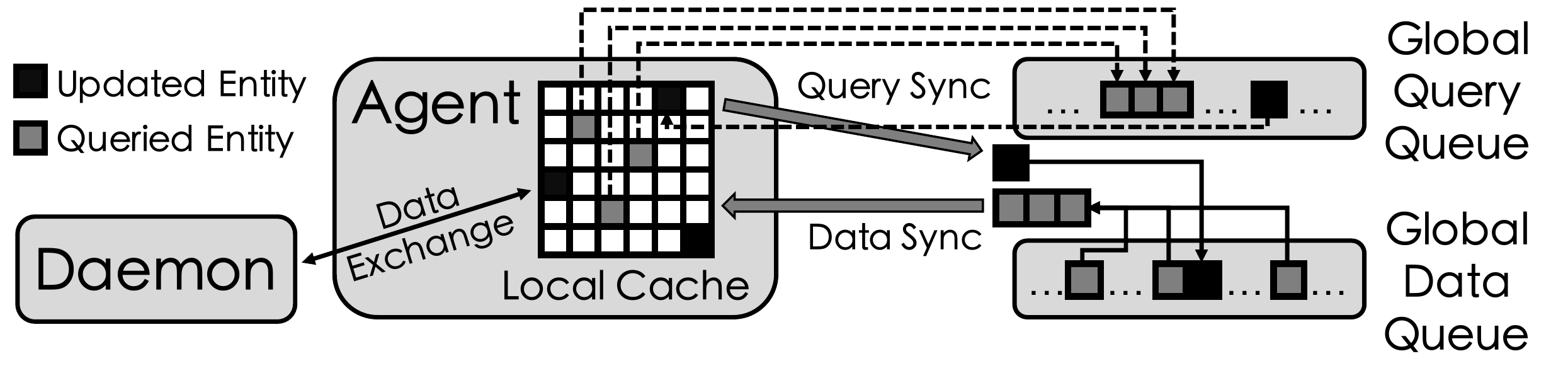}
    \vspace{-15pt}
    \caption{Synchronization Caching}
    \label{fig:SyncCachingOverview}
    \vspace{-5pt}
\end{figure}

Figure~\ref{fig:SyncCachingOverview} shows the main process of synchronization caching.
The idea is to use local cache of agents to reduce unnecessary data transferring between daemons and upper systems. It has two parts, \emph{LRU-based caching} and \emph{lazy uploading}.

\paragraph{LRU-based Caching}

Think twice about the data transferring process.
At the beginning of an iteration, an agent downloads data to be computed from upper systems.
A vertex would have to be repeatedly downloaded from upper systems, if it is involved in the computation iterations, even though its corresponding attributes are never updated.

To save the overhead, the agent can cache a set of vertices in a temporary vertex table, and the cache is organized in a least recently used (LRU {\it in short}) manner.
Initially, when entering the cache, every vertex has a weight, whose value decreases with the passage of iterations, and increases if being used for computation.
When the daemon requires a specific vertex for computation, the agent first checks its local cache for the vertex.
If not found, the agent downloads it from upper systems to cache, and evicts the vertex with the highest weight.
When the agent collects computation results for updating to upper systems, it first checks if corresponding vertices are cached. If so, the agent updates the attributes of the vertex, and upgrades its weight.
Otherwise, the agent chooses vertices with the lowest weights, and replaces them by vertices in the computation result.
If the chosen vertices were updated in previous iterations, corresponding information will be uploaded to upper systems.
The updated vertices in the cache are marked, for lazy uploading, as discussed below.


\paragraph{Lazy Uploading}

Also, there is no need for immediately uploading an updated vertex, until it is involved in the computation of other distributed nodes.
For example, if there are many copies of a vertex generated before the synchronization, only the vertex copy with optimal updated value needs to be uploaded, meaning that other vertex copies are obsolete.
Thus, to prevent unnecessary uploading, we make the strategy of ``lazy uploading''.

So, we design two queues for the lazy uploading, \emph{global query queue} and \emph{global data queue}. After all computation results are updated to the cache, the agent first constructs a list of vertex IDs which are needed by the distributed node for the next iteration.
Then, all agents push their local lists to the upper system. The union of local lists formulates the global query queue and is broadcast to all agents.
Each agent compares its cache with the global query queue, and uploads the required vertices to the global data queue.
This way, data uploading is triggered only if necessary.
Algorithm~\ref{al:SyncCaching} shows the details of lazy uploading.

\begin{algorithm}[ht]
    \small
    \caption{Lazy Uploading}
    \label{al:SyncCaching}
    \algorithmicrequire{Updated Dataset $s$, Global query queue $gqq$, Global data queue $gdq$}
    \begin{algorithmic}[1]
        \State $s_{q}$ $\leftarrow$ $s$.GetQueriedEntity()\;
        \State Send($gqq$, $s_{q}$)\;
        \State Wait for other agents\;
        \State $s_{u}$ $\leftarrow$ Find($gqq$, $s$.GetUpdatedEntity())\;
        \State Send($gdq$, $s_{u}$)\;
        \State Wait for other agents and upper system synchronization\;
        \State $s$.Update(Fetch($gdq$, $s_{q}$))\;
    \end{algorithmic}
\end{algorithm}

\subsubsection{Synchronization Skipping}

\label{subsubsec:syncskip}

Following the characteristics of distributed graph processing,
it happens that some iterations can be skipped, so that the synchronization overheads of these iterations can be saved.
The observation is that, there is no need to trigger the global synchronization, if there is no de facto ``conflicts'' among distributed nodes for being synchronized, i.e., updated data of a node is not needed by all other nodes.
We therefore design a mechanism called ``synchronization skipping'' based on synchronization caching, to detect if the current iteration synchronization can be skipped.
\begin{figure}[ht]
    \centering
    \vspace{-15pt}
    \includegraphics[width=\linewidth]{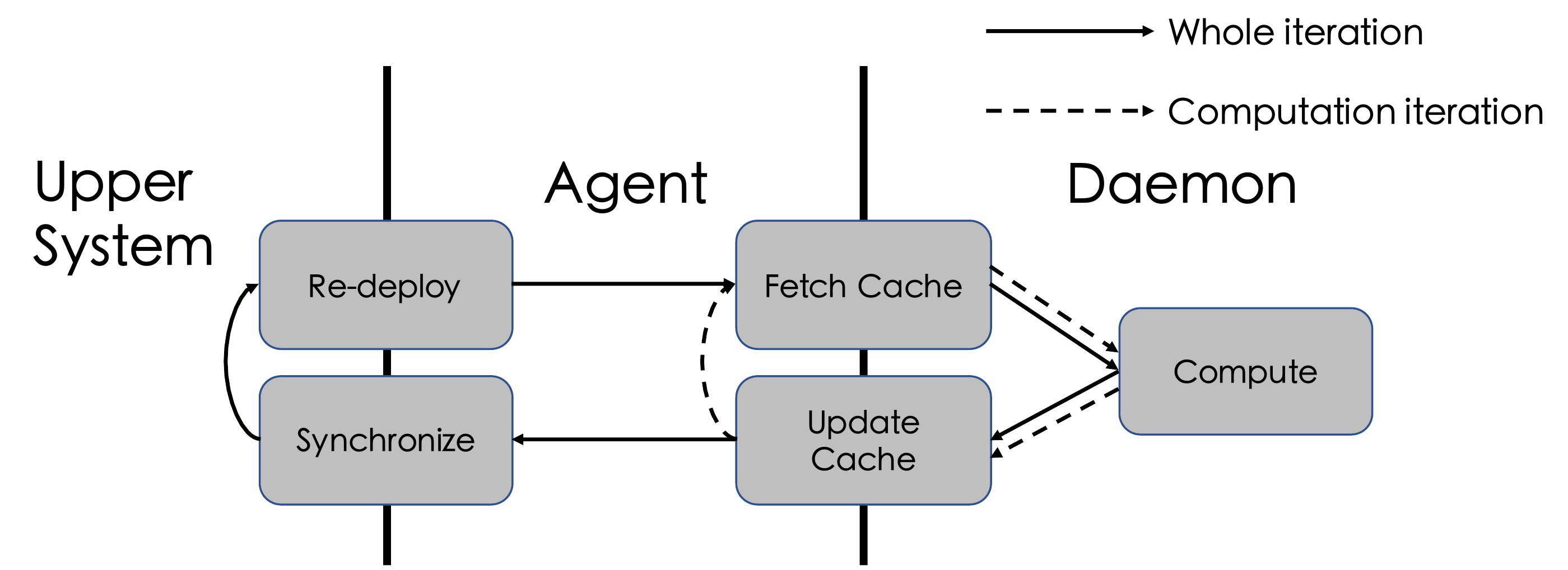}
    \vspace{-15pt}
    \caption{Synchronization Skipping}
    \label{fig:SyncSkippingOverview}
    \vspace{-5pt}
\end{figure} 

As shown in Figure~\ref{fig:SyncSkippingOverview}, at the end of cache updating, an agent checks if each updated vertex and its outer edges are in the same node.
If it is true for an agent, it means the agent can continue with the next iteration using its local data.
If it is true for all agents, it means that there is no need for any inter-node data transferring.
Thus, the upper system process can be skipped and next computation iteration can be directly started.
This way, multiple computation iterations can be equivalent to a logically combined iteration, and therefore unnecessary synchronization for intermediate iterations can be skipped. 

\subsection{Beyond-Iteration Optimization: Workload Balancing}

\label{subsec:WorkloadBalancing}

%
%
%



\subsubsection{Motivation}
As a ``software glue'', the middleware connects different accelerators and different upper systems. For instance, upper systems may adopt various graph partitioning strategies for assigning subgraphs to distributed nodes, which may cause storage imbalance. On the other hand, different accelerators may have different computation powers, which may cause computation imbalance. Therefore, it is important
for the middleware to have a mechanism to detect and react to the workload balancing, so that the performance of the parallelism can be maximized.

\subsubsection{Analysis}

%


To this end, we introduce a simple yet effective workload estimation model for the middleware to predict the performance of data processing of a local node.

Suppose there are in total $D$ data entities which are partitioned into $m$ distributed nodes, satisfying $\sum_{j=1}^m d_j = D$.
According to pipeline shuffle mechanism in Section~\ref{subsec:pipelineShuffle}, the total processing time consists of three parts.
For ease of discussion, for distribute node $j$, we set the total time cost $T^j_{total}$ taken by Thread.Download, Thread.Compute, and Thread.Upload as $T^j_{n}$, $T^j_{c}$, and $T^j_{u}$, respectively, satisfying $T^j_{total} = T^j_n + T^j_c + T^j_u$.
Noticed that both ${T^j_{n}}$ and ${T^j_{u}}$ are  proportional to $d_{j}$. ${T^j_{c}}$'s computation time and data copying time are also proportional to $d_{i}$, and calling time $T^j_{call}$ is proportional to the number of blocks $s$. We can have:
\begin{equation}
\small
    \label{eql:CostModelOneIteration}
    \nonumber
    \begin{aligned}
        T^j_{total} = {T^j_{n}} + {T^j_{c}} + {T^j_{u}} = c_{j}\cdot d_{j} + s \cdot T^j_{call},
    \end{aligned}
\end{equation}
where $c_{j}$ is the coefficient associated with node $j$ to represent the relation between data size and process time.
Since there is no relationship between $s$ and $d_{j}$, there is no need to consider $s$ in this situation.
Then, given a set of $m$ distributed nodes, the objective of workload balancing can be represented by:
\begin{equation}
\small
\label{eqn:minmax}
\min (\max\limits_{j \leq m}(c_{j}\cdot d_{j}))
\end{equation}

Here, we call $\frac{1}{c_{j}}$ as \textit{computation capacity factor}, since $c_{j}$ shows ``time cost per unit amount of data'', and $\frac{1}{c_{j}}$ shows ``data processed per unit time''.

\subsubsection{Mechanism}
With the objective function (Equation~\ref{eqn:minmax}), we come up with two metrics for the middleware to detect and react to the workload imbalance.
The first one offers benchmarks for upper systems to adjust partitioning strategies, given a specific configuration of accelerators for distributed nodes.
The second one offers benchmarks for upper systems to supervise the assignment of accelerators to distributed nodes, under a specific graph partitioning strategy.
In other words, our estimation model can be applied for $2$ cases.
The first case is on the tuning of $\{d_j\}$ with fixed $\{c_j\}$.
The second case is on the tuning of $\{c_j\}$ with fixed $\{d_j\}$.

{\it Case 1: tuning $\{d_j\}$ under fixed $\{c_j\}$.}
Lemma~\ref{lemma:minmaxcidi} makes the theoretical basis for getting optimal values of $d_{j}$.

\begin{lemma}
    \label{lemma:minmaxcidi}
    \textnormal{Given $D$ data entities which are partitioned to $m$ distributed nodes, where each node holds a data fragment $d_j$, satisfying $\sum_{j = 1...m} d_{j} = D$, the balancing target is to minimize function $G(.)$, which represents the maximum time cost of a distributed node.}
    \begin{equation}
    \small
        \nonumber
        \begin{aligned}
            G(d_{1}, ..., d_{m}) = \max_{j = 1...m}(c_{j}d_{j})
        \end{aligned}
    \end{equation}
    \textnormal{It holds that function $G(.)$ achieves its minimum value, iff every element $d_j$ of its $m$-dimensional input variable $\{d_{j}\}_{j\leq m}$ satisfies: }
    \begin{equation}
    \vspace{-3pt}
    \small
        \nonumber
        \begin{aligned}
            d_{j} = \frac{\frac{1}{c_{j}}}{\sum^{m}_{j = 1} \frac{1}{c_{j}}}D
        \end{aligned}
    \end{equation}
\end{lemma}

\begin{proof}
    First, if every $d_{j}$ meets the condition, we have:
    \begin{small}
    \begin{equation}
        \nonumber
        \begin{aligned}
            G(d_{1}, ..., d_{m}) = \max^m_{j = 1} \{c_{j} \cdot \frac{\frac{1}{c_{j}}}{\sum^{m}_{j = 1} \frac{1}{c_{j}}}D\} = \frac{D}{\sum^{m}_{j = 1} \frac{1}{c_{j}}}
        \end{aligned}
    \end{equation}
    \end{small}
    Second, we prove that for any possible $d_{j}$, we have $G \geq \frac{D}{\sum^{m}_{j = 1} \frac{1}{c_{j}}}$.
    We prove this assertion by contradiction. We first assume it holds that $G = max^m_{j = 1}(c_{j}d_{j}) < \frac{D}{\sum^{m}_{j = 1} \frac{1}{c_{j}}}$.
    Then for every $d_{j}$, we have:
    \begin{equation}
    \small
        \nonumber
        \begin{aligned}
            d_{j} < \frac{\frac{1}{c_{j}}}{\sum^{m}_{j = 1} \frac{1}{c_{j}}}D \Rightarrow D = \sum^m_{j = 1} d_{j} < \frac{\sum^m_{j = 1} \frac{1}{c_{j}}}{\sum^m_{j = 1} \frac{1}{c_{j}}}D = D
        \end{aligned}
    \end{equation}

    Here, contradiction occurs. Thus, function $F$ reaches the minimum value $\frac{D}{\sum^{m}_{j = 1} \frac{1}{c_{j}}}$, if and only if for all $d_{j}$, $d_{j} = \frac{\frac{1}{c_{j}}}{\sum^{m}_{j = 1} \frac{1}{c_{j}}}D$. The lemma is hence proved.
\end{proof}

Lemma~\ref{lemma:minmaxcidi} shows that $\frac{\frac{1}{c_{j}}}{\sum^{m}_{j = 1} \frac{1}{c_{j}}}$ can serve as \textit{balancing factor}s for selecting appropriate partitioning strategies.
For example, given a set of partitioning strategies, the one that achieves minimum $F(.)$ is to be selected.

{\it Case 2: tuning $\{c_j\}$ under fixed $\{d_j\}$.}
It is possible for upper systems to elastically select demanding
number of accelerators (e.g., from a GPU cloud), given that
the graph partitioning results are fixed (graph partitioning is more I/O and computational intensive than other processing phases). If so, the middleware can adjust the \textit{computation capacity factor} 
$1/c^{'}_{j}$
for
balancing the workload, according to Lemma \ref{lemma:minmaxdici}.

\begin{lemma}
    \label{lemma:minmaxdici}
    \textnormal{Given $D$ data entities which are partitioned to $m$ distributed nodes, where each node holds a data fragment $d^{}_j$, satisfying $\sum^{m}_{j = 1} d^{}_{j} = D$, and given the maximum available computation capacity factor $f$ ($f \geq \max_{j = 1...m}\frac{1}{c^{}_{j}}$),
    our target is to minimize function $G'(.)$, which indicates the maximum time cost of a distributed node.}
    \begin{equation}
    \small
        \nonumber
        \begin{aligned}
            G^{'}(\frac{1}{c^{}_{1}}, ..., \frac{1}{c^{}_{m}}) = \max^m_{j = 1}(c^{}_{j}d^{}_{j})
        \end{aligned}
    \end{equation}

    \textnormal{Function $G'(.)$ achieves its minimum value, if every element $\frac{1}{c_j}$ of its $m$-dimensional input variable $\{\frac{1}{c_{j}}\}_{j\leq m}$ satisfies:}
    \begin{equation}
     \small
        \nonumber
        \begin{aligned}
            \frac{1}{c^{}_{j}} = \frac{f \cdot d^{}_{j}}{d_*}
        \end{aligned}
        \textit{,~~~~~~~where $d_* = max_{j\leq m}(d_j)$}
    \end{equation}
\end{lemma}

\begin{proof}
    Let $\frac{1}{c_*}$ be $max_{j\leq m}\frac{1}{c_j}$. Since $\frac{1}{c_*} \leq f$, we have:
    \begin{equation}
     \small   \nonumber
        \begin{aligned}
            \frac{d^{}_*}{f} \leq c^{}_*d^{}_* \leq F^{'} = \max^m_{j = 1}(c^{}_{j}d^{}_{j})
        \end{aligned}
    \end{equation}

    To make $G^{'} = \frac{d^{}_*}{f}$, all other $c^{}_{j}d^{}_{j}$ must be not greater than $\frac{d^{}_*}{f}$. Thus, to minimize $\frac{1}{c^{}_{j}}$, we have:
    \begin{equation}
     \small   \nonumber
        \begin{aligned}
            \frac{1}{c^{}_{j}} = min\{\frac{1}{c^{}_{j}}, \textnormal{where }c^{}_{j}d^{}_{j} \leq \frac{d^{}_*}{f}\} = \frac{fd^{}_{j}}{d^{}_*}
        \end{aligned}
    \end{equation}
Thus, the lemma is proved.
\end{proof}

According to Lemma~\ref{lemma:minmaxdici}, the middleware can dynamically allocate idle accelerators to generate more daemons for the node demanding computation powers,
as long as conditions of computation capacity factor of each partition are met.

\section{System Implementation}
\label{sec:si}

We show details on key implementation of the middleware.

\subsection{APIs}

\label{subsec:API}

It is important for the middleware to create a series of easy-to-use interfaces to make accelerators plugged to upper systems easy and coder-friendly.
In our implementation, we design an iteration-based graph algorithm template and a set of operations interfaces to connect with upper systems.



\subsubsection{Algorithm Template}
\label{subsubsec:template}

The APIs of algorithm template follow a unified iterative model and support C++-based code integration, including OpenMP, OpenCL, MPI, and CUDA.
There are $3$ steps with computation of an iteration for general multiworker systems, \textit{Message Passing}, \textit{Combining} and \textit{Aggregating}, in which external computation resources can be utilized for computation optimization.
In correspondence to the above $3$ steps, our algorithm template has $3$ APIs, MSGGen(), MSGMerge() and MSGApply().
MSGGen() function is the computation function for calculating the initial results with vertex and edge blocks and transforming them into initial messages. MSGMerge() function delivers the initial messages to corresponding graph partitions. MSGApply() fetchs message sets for the current partition, and applies them to corresponding vertices and edges. Accordingly, one can design a graph algorithm by implementing the 3 interfaces of the algorithm template.
Examples of implementing graph algorithms can be found in our code repository$^{\ref{source}}$.

With the help of the daemon-agent framework, runtime details, such as data transferring, runtime orders, interactions with upper systems, and extra resource management, are hidden to algorithm engineers
, thus they can focus on the implementation of three APIs for specific graph algorithms.

More, with the separate maintenance of thee functions,
upper system developers can arrange the API calls in different orders, so that the middleware is adaptable to various graph computation models, such BSP, GAS, and asynchronous model, as shown in Section~\ref{subsubsec:powergraph}.



\subsubsection{Operation Interfaces}
\label{subsubsec:interface}

To make upper system calls easier to be adapted to agents, the agent accessing is organized into three functions,
including two functions for data transferring, i.e., transfer() and update(), and function requestX() for computation lifecycle controlling. Here, $X$ can be any of the three APIs, MSGGen(), MSGMerge() or MSGApply().
For an upper system, a call sequence of a computation iteration is: connect() $\rightarrow$ update() $\rightarrow$ \{requestX()\} $\rightarrow$ update() $\rightarrow$ disconnect().
Upper system developers only need to access corresponding functions by inputting proper parameters to get the full control of the daemon runtime. Also, it takes merely a few lines of code for the agent to connect to upper systems.

In summary, with such interface functions of agents, computation daemons can be integrated and cooperated for the global computation invoked by upper systems.



\subsection{Environment Accessing}
\label{subsec:env}

\subsubsection{GraphX (JVM)
}


JVM is a uniform environment separated from the local environment of distributed nodes to execute Java programs.
However, it makes things complicated when GraphX needs external tools or libraries for computation.
JNI provided in JVM suffers from additional costs in invoking native target functions, due to JNI callbacks, which can be eased by the native memory \cite{DBLP:journals/corr/HalliCM14}.
To solve the problem, we design $2$ components to efficiently break the barrier between the JVM runtime and local environment.

{\bf JNI Transmitter.}
We use JNI to achieve reflection between GraphX native method interfaces and external functions connected with agent.
However, naively invoking JVM methods at runtime incurs significant transmission lags.
Hence, we utilize a series of techniques such as POSIX-based shared memory data exchanging, batch data transferring, in JNI transmitter in order to reduce JNI calls.

{\bf Data Packager.}
Data packager solves the inconsistency of data structures of the two ends of JVM and local environment.
It uses techniques of bit data organization and space reserving for data transformation without extra space usage and redundant data copying.
Preliminary testing shows that about $3$ to $10$ times of improvement can be achieved, compared to direct target function invoking.

\subsubsection{PowerGraph
}

\label{subsubsec:powergraph}

PowerGraph\cite{DBLP:conf/osdi/GonzalezLGBG12} follows another computation model, called Gather-Apply-Scatter (GAS {\it in short}), which is widely used by many distributed graph processing frameworks.
Although BSP and GAS are of different graph computation models, they are common in basic iterative characteristics~\cite{DBLP:journals/pvldb/Ozsu2019Keynote}, which can be view as different orders of iterative operations, etc. It paves the road for theoretical and technical foundation of a general middleware design in supporting different computation models.

For example, when connecting to PowerGraph, MSGGen(), MSGMerge() and MSGApply() is used to represent scatter, gather, and apply steps in GAS model.
The execution order in PowerGraph thus follows Merge() $\rightarrow$ Apply() $\rightarrow$ Gen(), which differs from BSP model, i.e., Gen() $\rightarrow$ Merge() $\rightarrow$ Apply().
In an iteration, PowerGraph calls agent interfaces, in the order of requestMerge(), requestApply() and requestGen(), so that GAS model can be supported by the middleware without any extra code modifications. It shows the generality of our middleware in adapting to different computation models.

\subsection{Runtime Isolation}
\label{subsec:RuntimeIsolation}


If an agent is na\"ively designed as a parent process of daemons, the device environment of accelerators would be re-initialized multiple times during the iterative graph processing. Because the launching and ceasing of an agent must be triggered multiple times by the upper system, and so do their associated daemons.
The frequent re-initialization incurs considerable system overheads, since the initialization process of daemons (with internal function calls) and associated computational devices is time-consuming.
To overcome the dilemma, our daemon-agent framework detaches the initialization process from direct function calls.
Daemons and agents work as independent processes, and they communicate with each other by message exchange.
This way, a daemon never triggers re-initialization during the iterative graph processing.

\section{Evaluation}
\label{sec:ev}

\begin{figure*}[ht!]
\centering
\vspace{-5pt}
\begin{subfigure}[t]{0.245\linewidth}
    \centering
    \includegraphics[width=\linewidth]{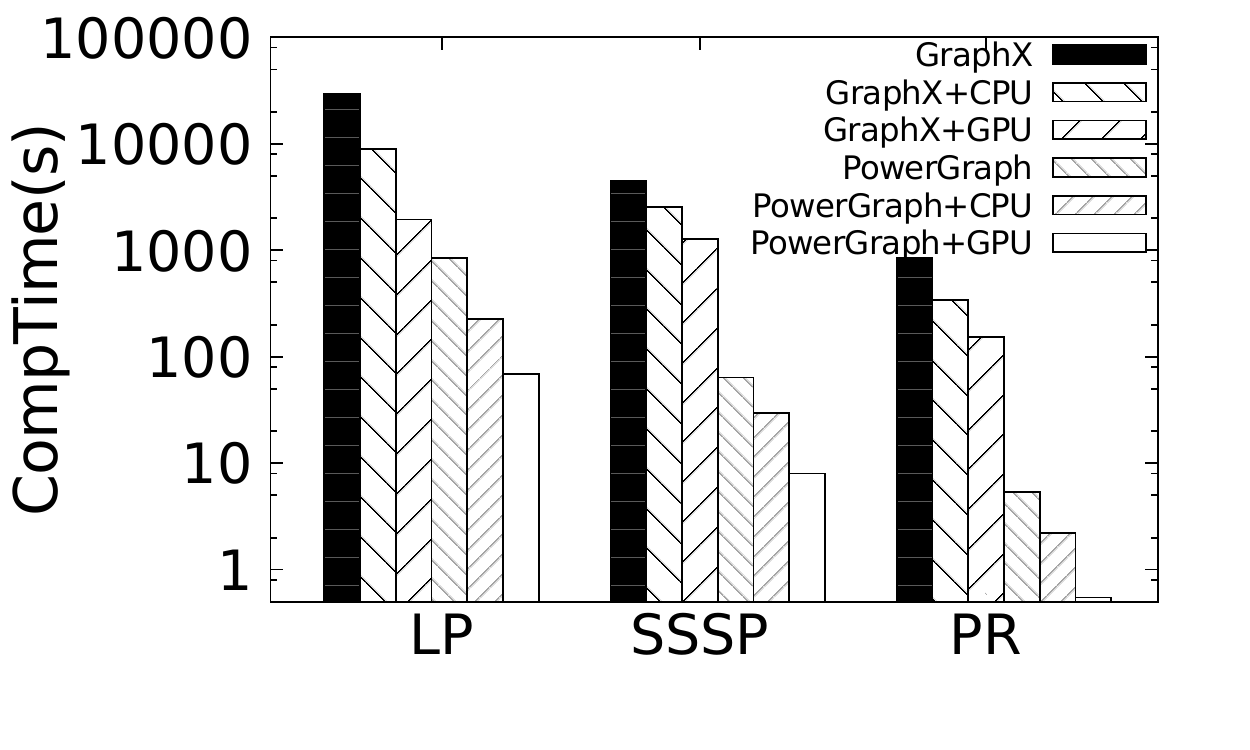}
    \vspace{-25pt}
    \caption{algos @ Twitter2010}
    \label{fig:EXP_GraphXTwitter2010}
\end{subfigure}
\begin{subfigure}[t]{0.245\linewidth}
    \centering
    \includegraphics[width=\linewidth]{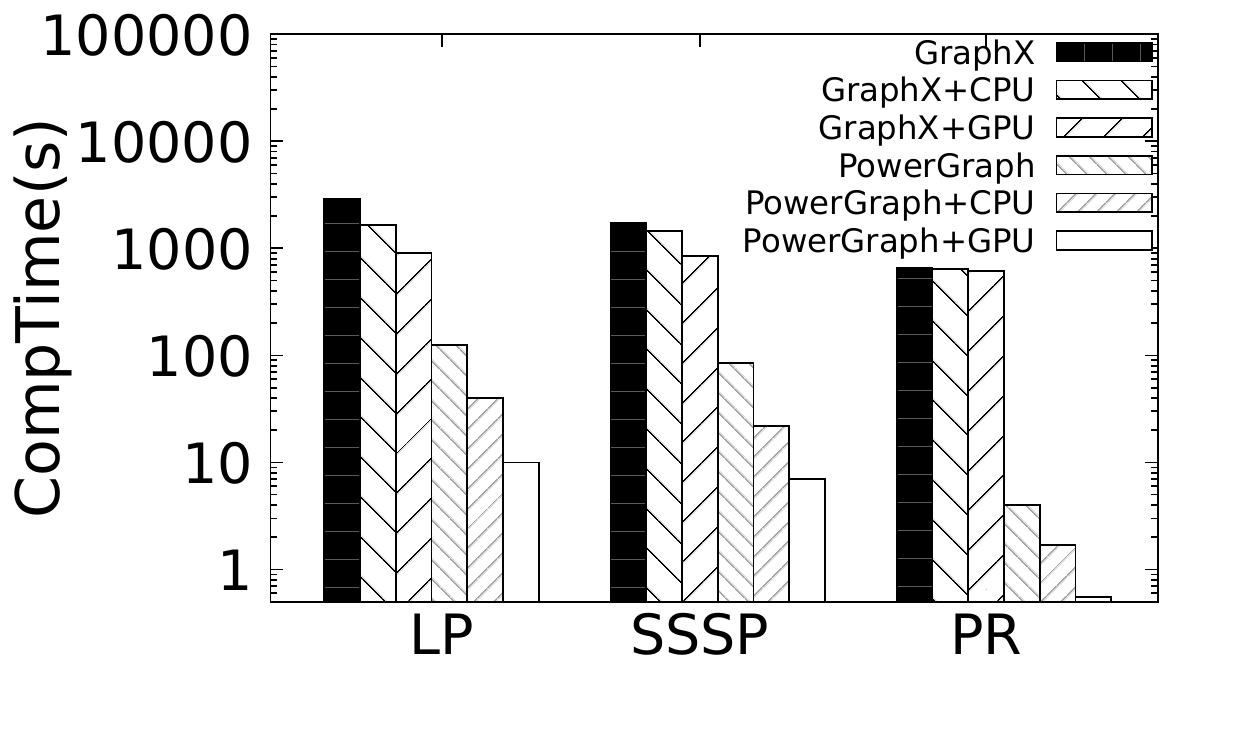}
    \vspace{-25pt}
    \caption{algos @ Orkut}
    \label{fig:EXP_GraphXOrkut}
\end{subfigure}
\begin{subfigure}[t]{0.245\linewidth}
    \centering
    \includegraphics[width=\linewidth]{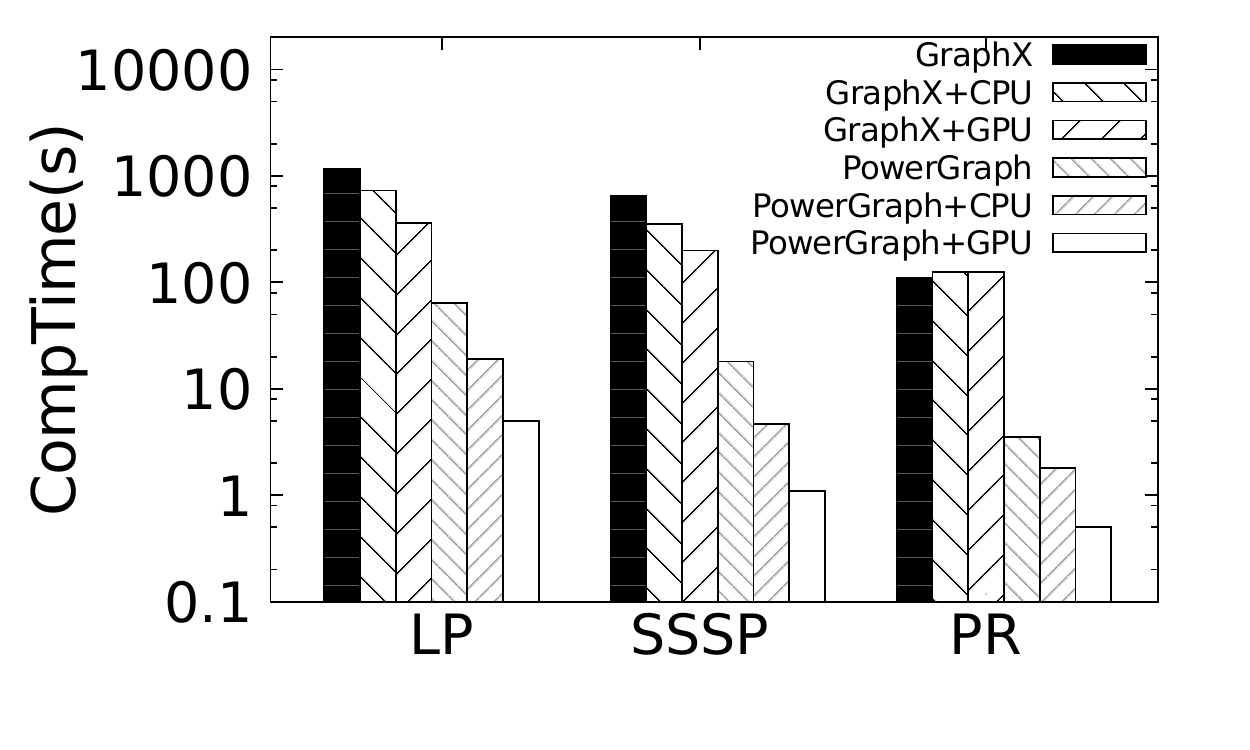}
    \vspace{-25pt}
    \caption{algos @ LiveJournal}
    \label{fig:EXP_GraphXLiveJournal}
\end{subfigure} 
\begin{subfigure}[t]{0.245\linewidth}
    \centering
    \includegraphics[width=\linewidth]{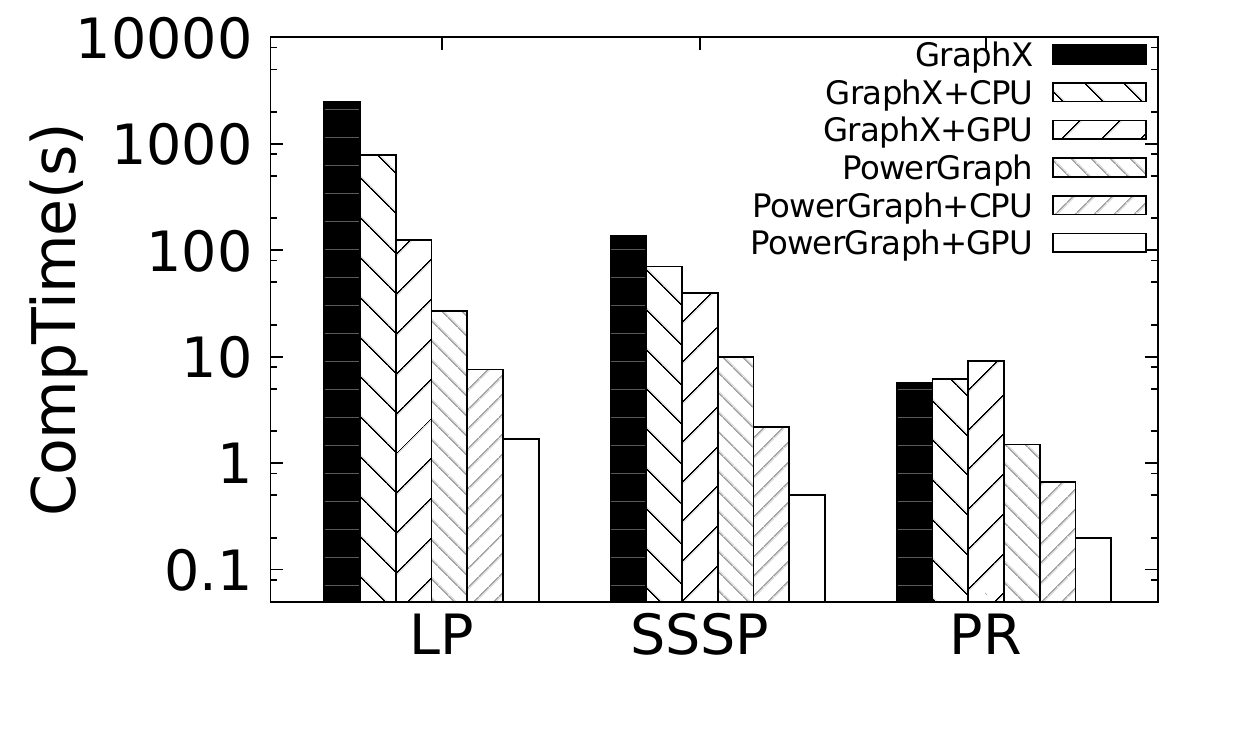}
    \vspace{-25pt}
    \caption{algos @ Wiki-topcats}
    \label{fig:EXP_GraphXWiki-topcats}
\end{subfigure}
\vspace{-5pt}
\caption{Results on Different Datasets}
\label{fig:EXP_BenchmarkReal}

\centering
\begin{subfigure}[t]{0.245\linewidth}
    \centering
    \includegraphics[width=\linewidth]{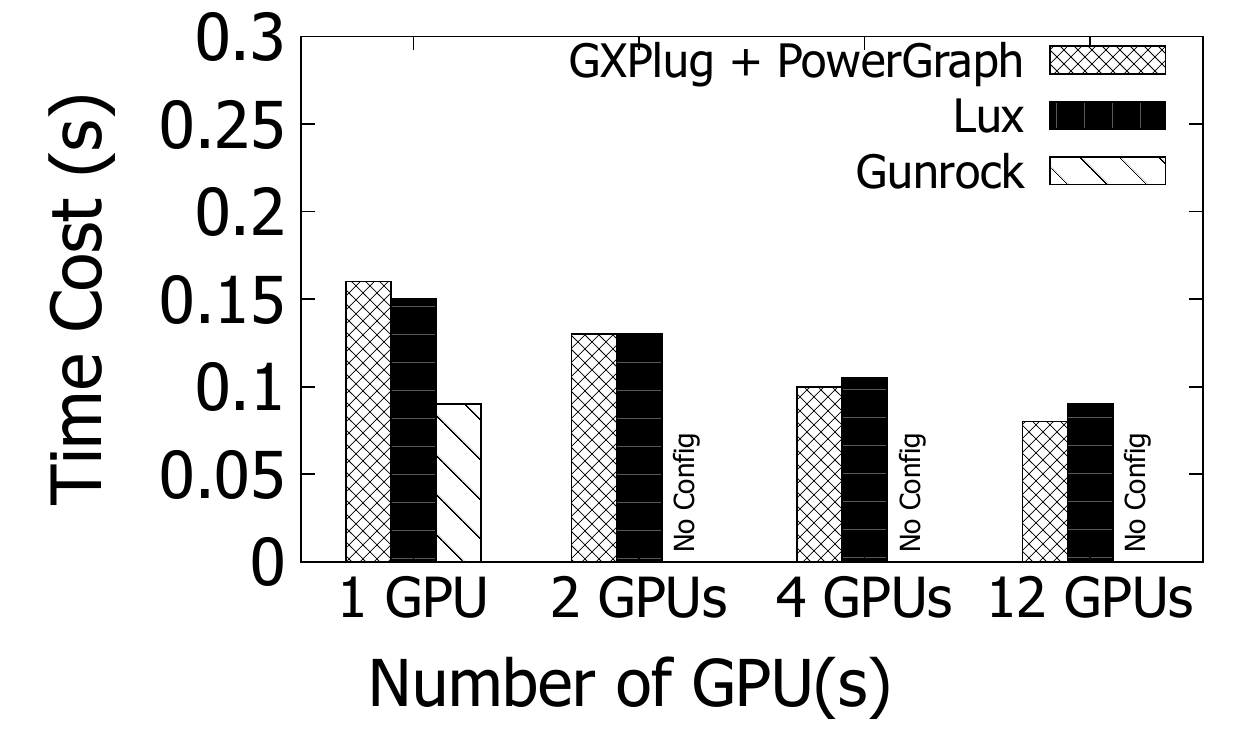}
    \vspace{-10pt}
    \caption{Scalability w.r.t. GPUs. (Orkut)}
    \label{fig:EXP_nodeElasticityOrkut}
\end{subfigure}
\begin{subfigure}[t]{0.245\linewidth}
    \centering
    \includegraphics[width=\linewidth]{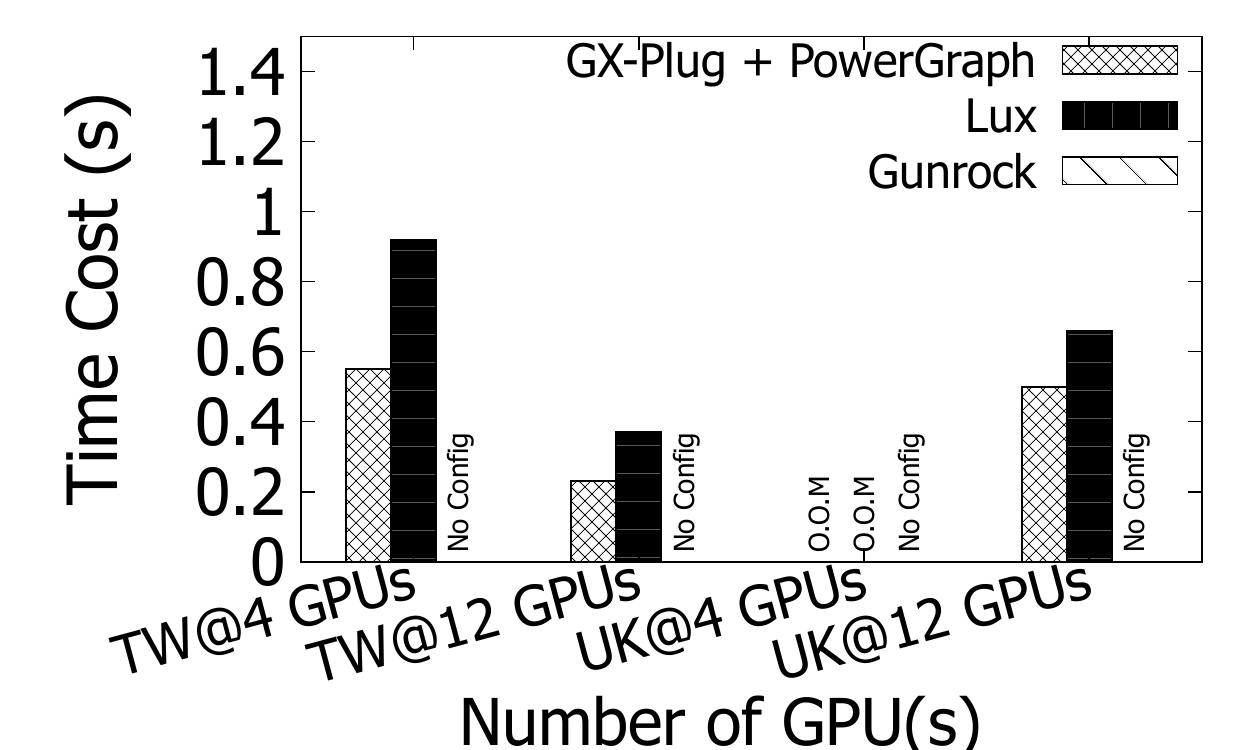}
    \vspace{-10pt}
    \caption{Results on Twitter \& UK-2007} 
    \label{fig:EXP_nodeElasticityLarge}
\end{subfigure}
\begin{subfigure}[t]{0.245\linewidth}
    \centering
    \includegraphics[width=\linewidth]{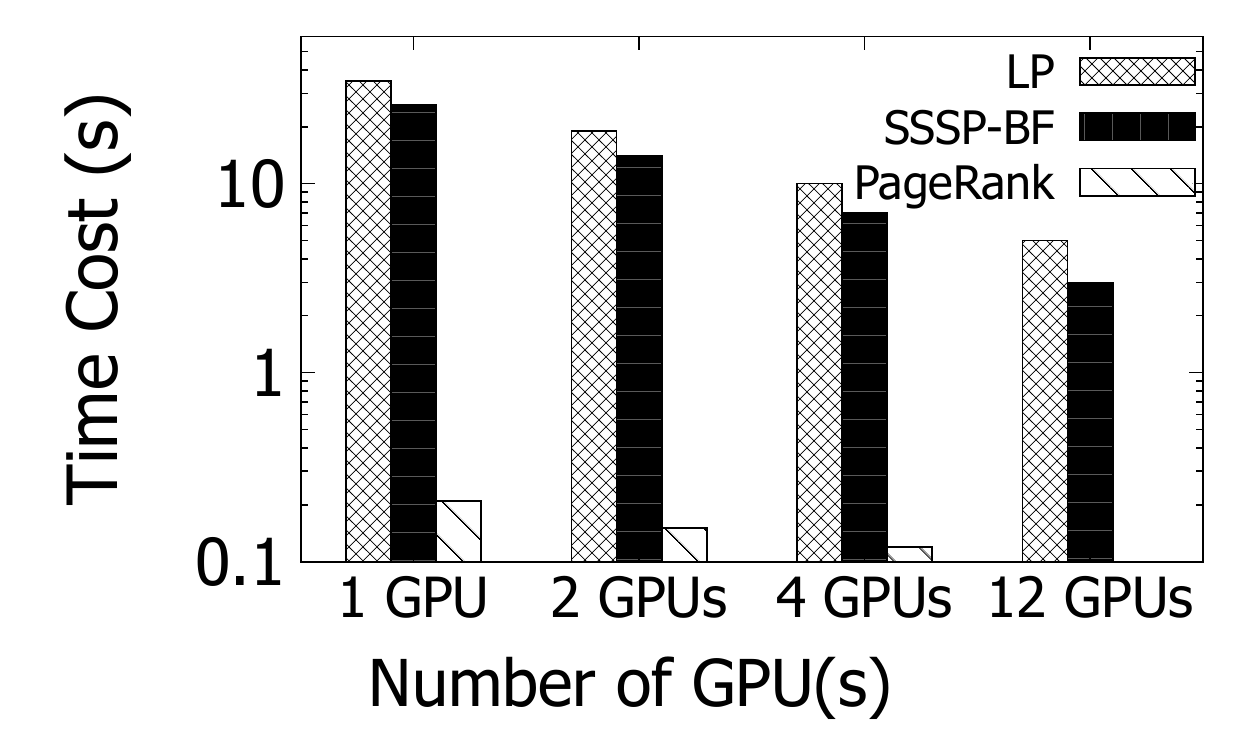}
    \vspace{-10pt}
    \caption{Scalability w.r.t. Algos}
    \label{fig:EXP_nodeElasticityAlgo}
\end{subfigure}
\begin{subfigure}[t]{0.245\linewidth}
    \centering
    \includegraphics[width=\linewidth]{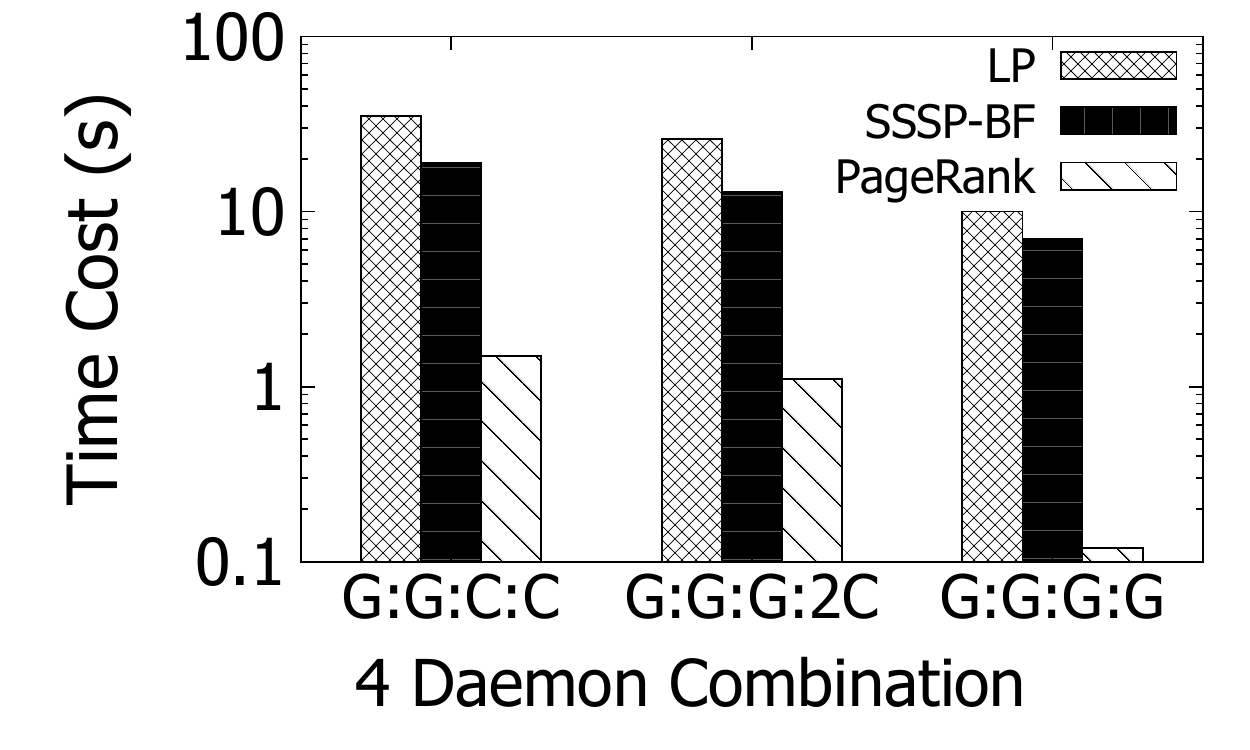}
    \vspace{-10pt}
    \caption{Mix \& Match (CPU and GPU)}
    \label{fig:EXP_nodeElasticityMixedDeployment}
\end{subfigure}
\caption{Results on Scalability}
\label{fig:EXP_BenchmarkCapacityScale}




\centering
\begin{minipage}[t]{0.195\linewidth}
    \centering
    \includegraphics[width=\linewidth]{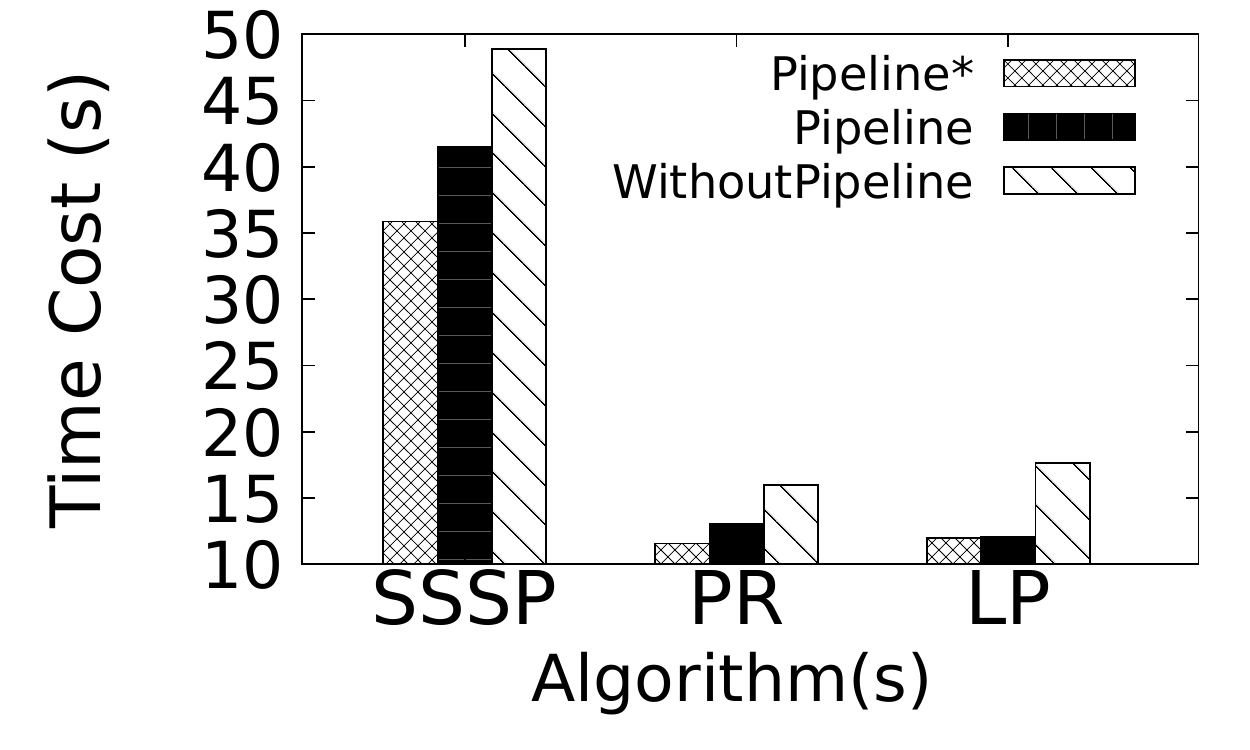}
    \caption{Performance of Pipeline Shuffle}
    \label{fig:EXP_PipelineShuffle}
\end{minipage}
\begin{minipage}[t]{0.39\linewidth}
    \centering
    \begin{subfigure}[t]{0.49\linewidth}
        \centering
        \includegraphics[width=\linewidth]{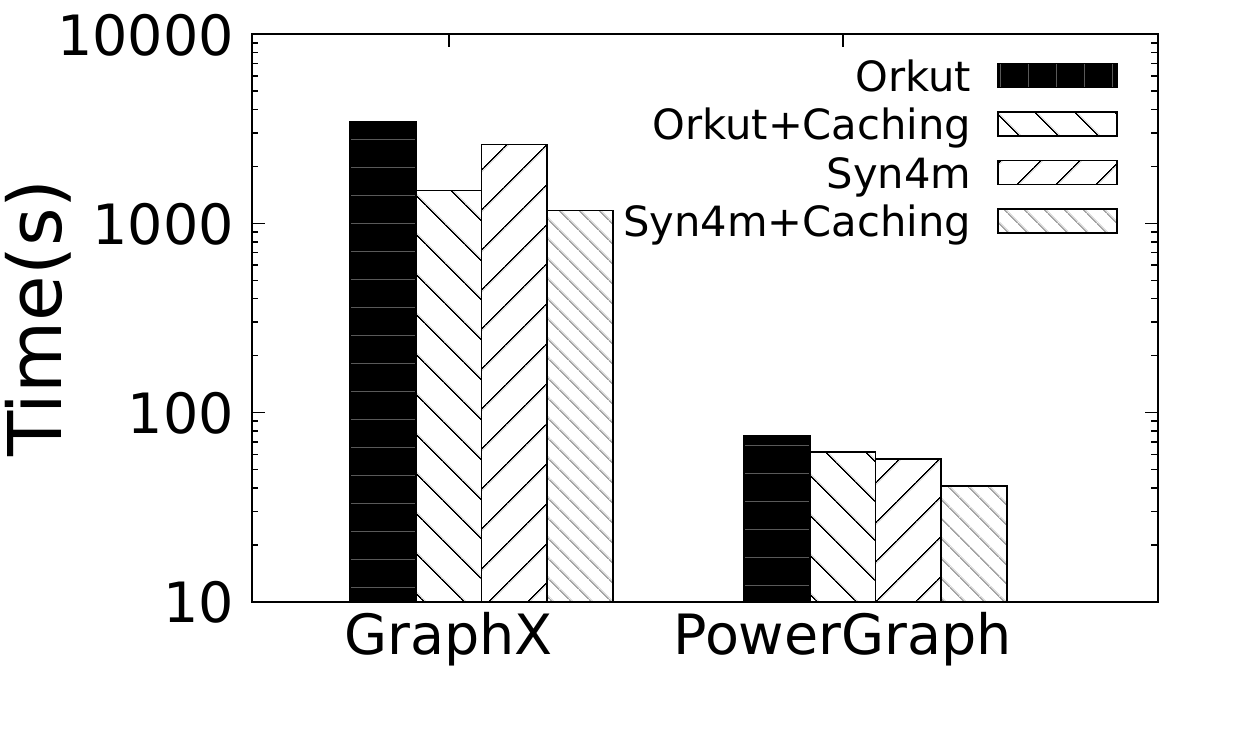}
        \caption{Caching on Systems}
        \label{fig:EXP_SynchronizationCachingPowerGraph}
    \end{subfigure}
    \begin{subfigure}[t]{0.49\linewidth}
        \centering
        \includegraphics[width=\linewidth]{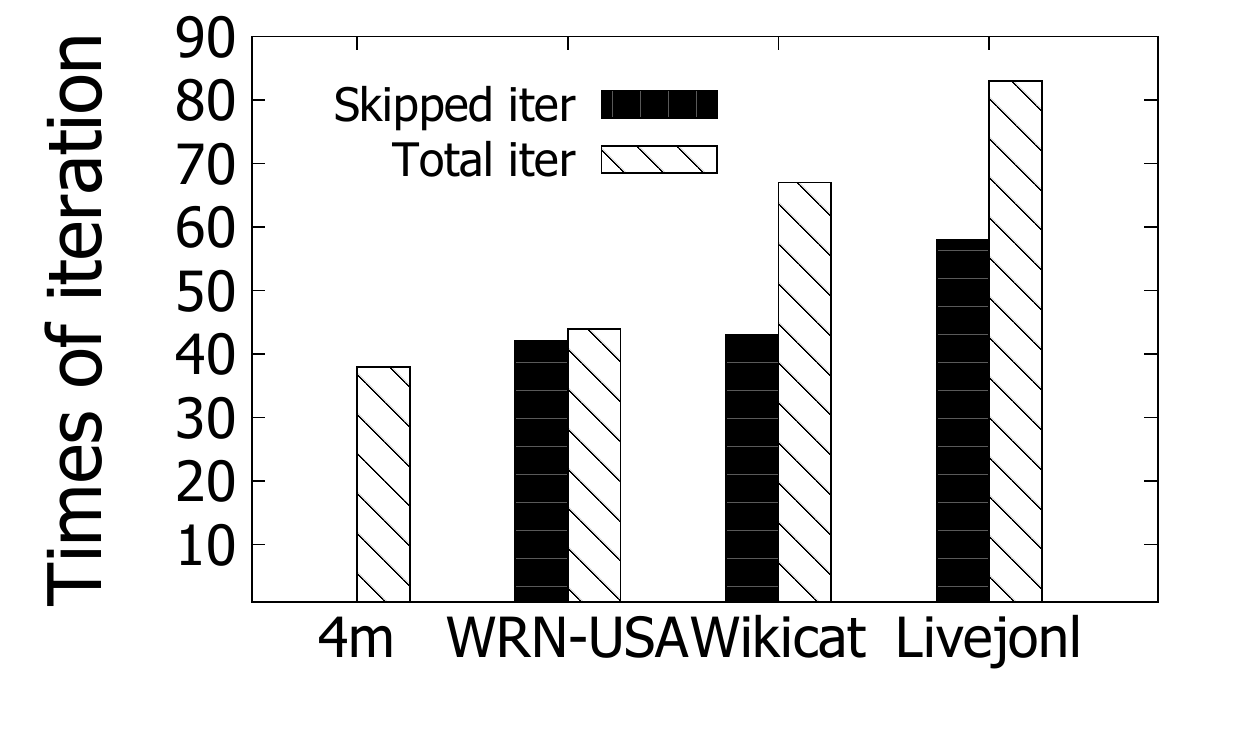}
        \caption{\# of Skipped Iterations}
        \label{fig:EXP_SynchronizationSkippingIters}
    \end{subfigure}
    \caption{Sync Cacheing \& Skipping Performance}
    \label{fig:EXP_SynchronizationCachingSkipping}
\end{minipage} 
\begin{minipage}[t]{0.39\linewidth}
    \centering
    \begin{subfigure}[t]{0.49\linewidth}
        \centering
        \includegraphics[width=\linewidth]{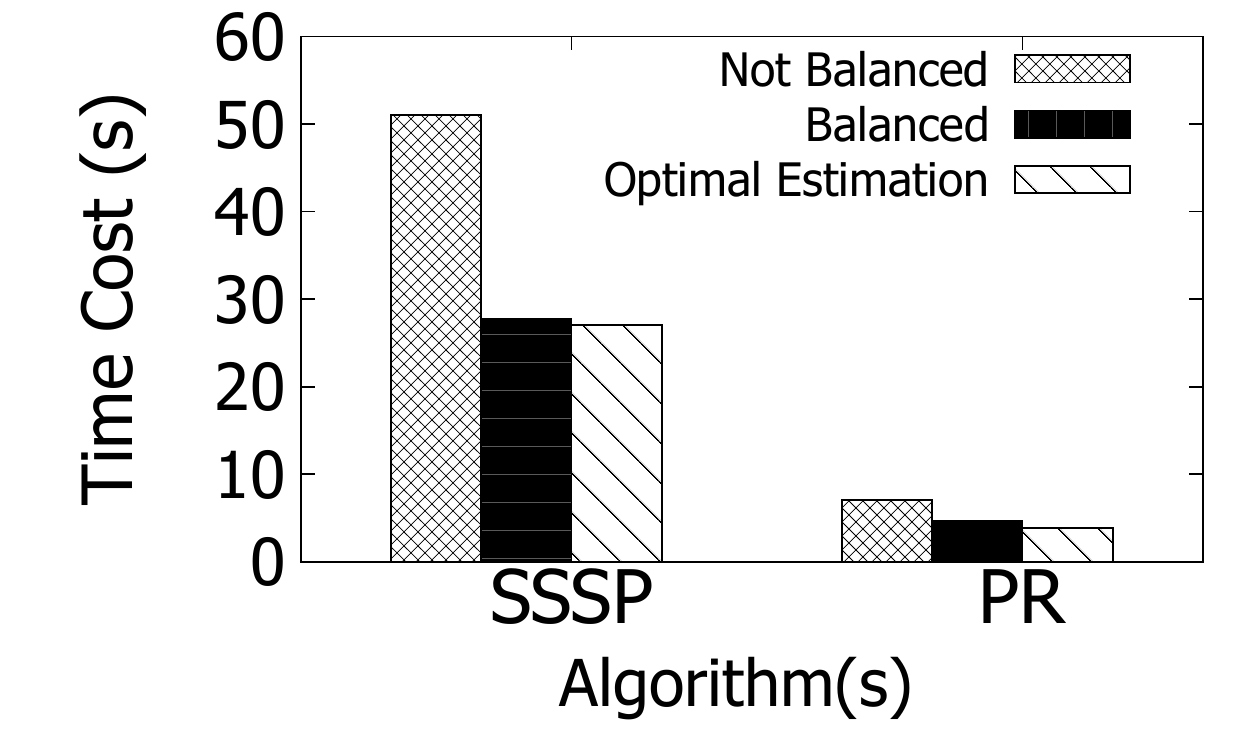}
        \caption{Fixed Comp Resource}
        \label{fig:EXP_ComputationResourceFixed}
    \end{subfigure}
    \begin{subfigure}[t]{0.49\linewidth}
        \centering
        \includegraphics[width=\linewidth]{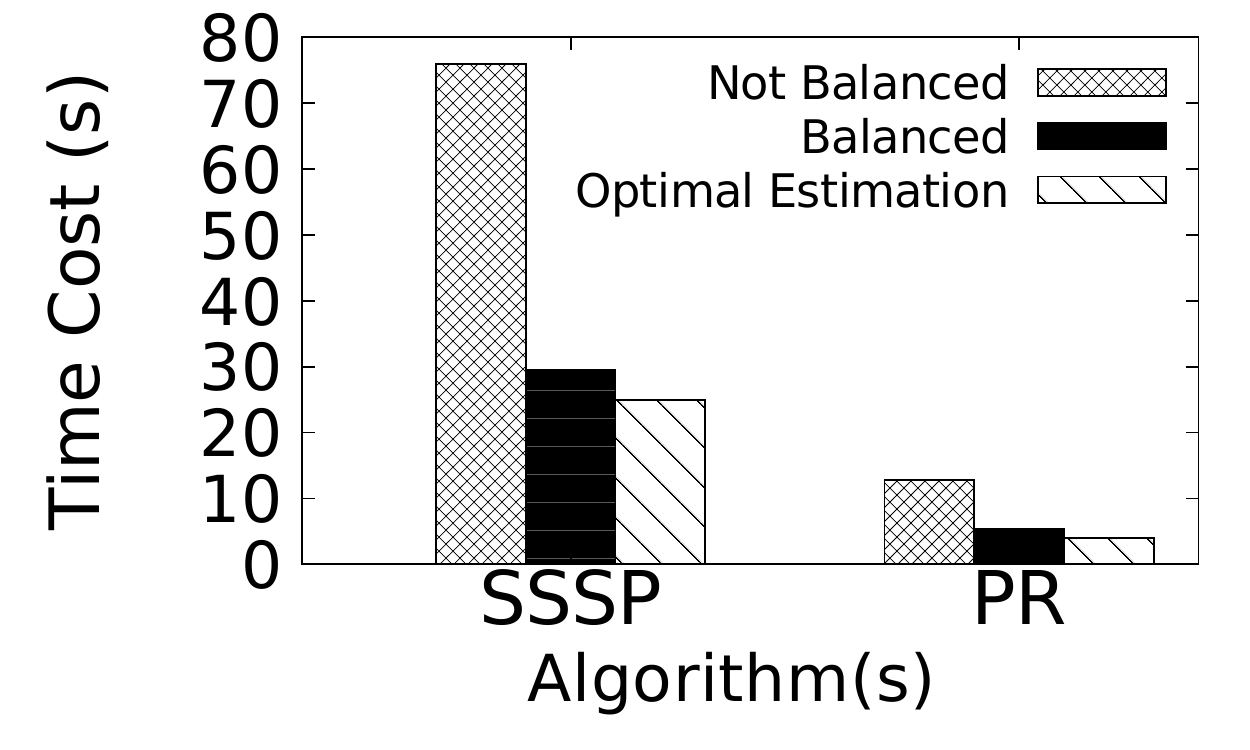}
        \caption{Fixed Data Partitioning}
        \label{fig:EXP_DataSizeFixed}
    \end{subfigure}
    \caption{Workload Balancing Performance}
    \label{fig:EXP_WorkloadBalancing}
\end{minipage}


\begin{minipage}[t]{0.24\linewidth}
    \centering
    \begin{subfigure}[t]{\linewidth}
        \centering
        \includegraphics[width=\linewidth]{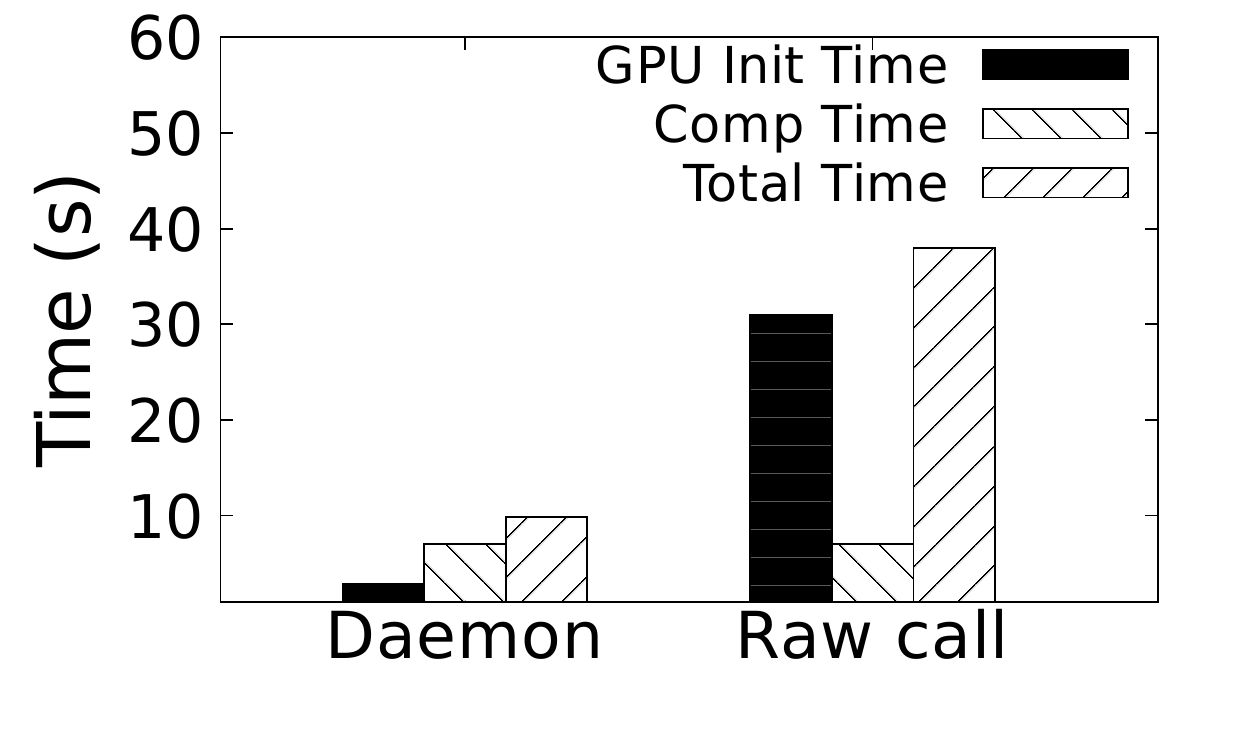}
        \label{fig:EXP_GPUInitSynOrkut}
    \end{subfigure}
    \caption{Runtime Isolation}
    \label{fig:EXP_GPUInit}
\end{minipage}
\begin{minipage}[t]{0.49\linewidth}
    \centering
    \begin{subfigure}[t]{0.49\linewidth}
        \centering
        \includegraphics[width=\linewidth]{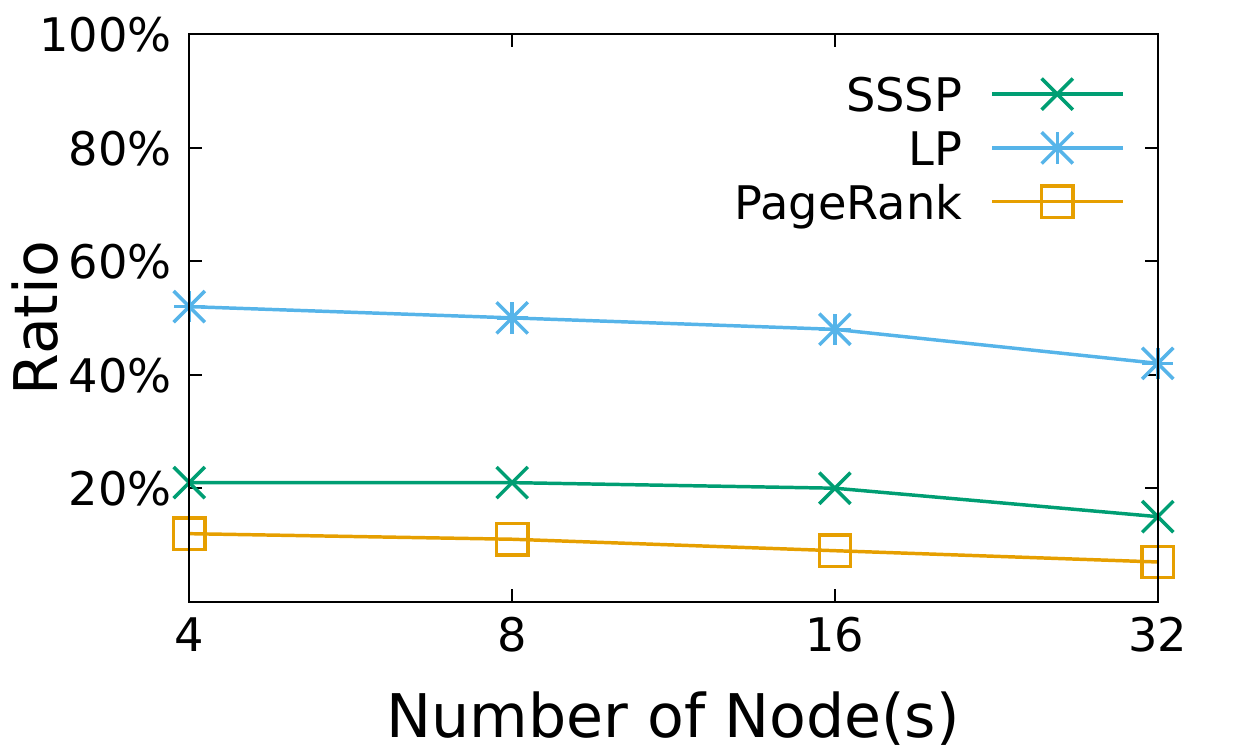}
        \caption{PowerGraph, Orkut}
        \label{fig:EXP_MidCost_PowerGraph} 
    \end{subfigure}
    \begin{subfigure}[t]{0.49\linewidth}
        \centering
        \includegraphics[width=\linewidth]{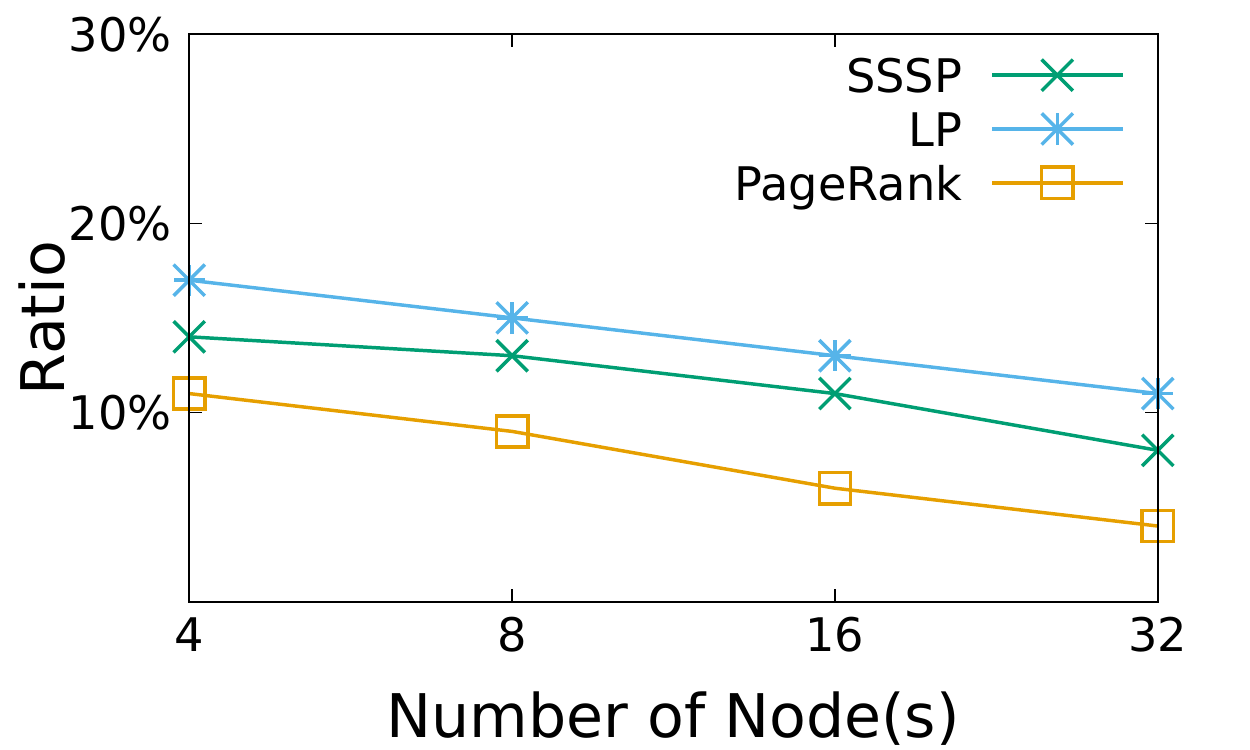}
        \caption{GraphX, Orkut}
        \label{fig:EXP_MidCost_GraphX}
    \end{subfigure}
    \caption{Middleware Cost Ratio}
    \label{fig:EXP_MidCost}
\end{minipage}
\begin{minipage}[t]{0.24\linewidth}
    \centering
    \begin{subfigure}[t]{\linewidth}
        \centering
        \includegraphics[width=\linewidth]{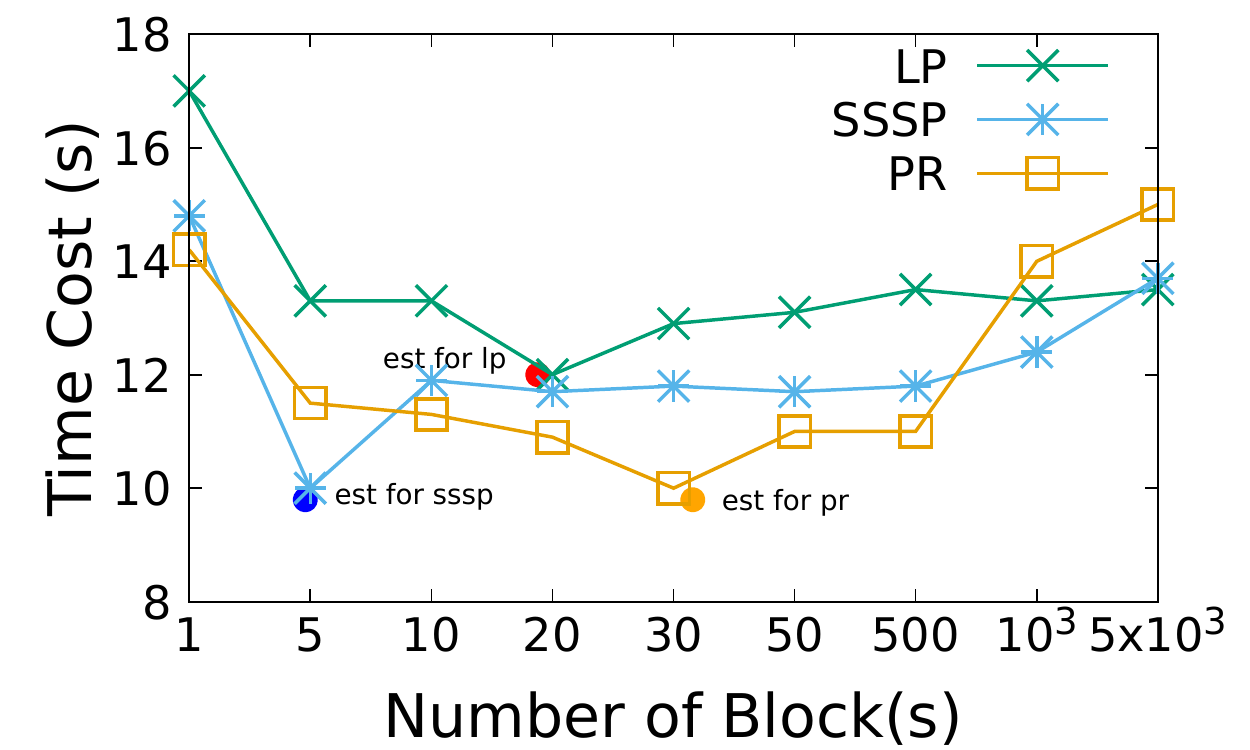}
        \caption{Estimation vs Real}
    \end{subfigure}
    \caption{Estimating $s_{opt}$}
    \label{fig:EXP_PipelineShuffleBlockSizeSelection}
\end{minipage}
\vspace{-10pt}
\end{figure*}

\subsection{Setup}
\label{subsec:setup}

We conduct experiments on a set of representative graph algorithms, such as Bellman-Ford (SSSP-BF), PageRank (PR), and Label Propagation algorithm (LP), by varying datasets of different distributions and scales\footnote{For SSSP-BF, we use 4 vertices as source vertices and calculate their SSSPs simultaneously to make it more compute-intensive. For LP, we limit the iterations to 15 times to avoid unlimited computation on specific datasets.}.
For experiments, we use a series of $6$ real datasets which are commonly used for graph systems testing,
as shown in Table \ref{tbl:dataset}. 
By default, Orkut is used, since it has the highest vertex degree among the $6$\footnote{The workload of a distributed node is proportional to the number of edges stored in it \cite{nature, DBLP:conf/osdi/GonzalezLGBG12}. Accordingly, a dataset with high vertex degrees offers high computation workload per unit amount of data.}.

\begin{table}[h!]
\footnotesize
\caption{Datasets}
\label{tbl:dataset}
  \centering
  \begin{tabular}{|r|r|r|r|}
    \hline
    Dataset & Vertex & Edge & Type\\
    \hline
    Orkut\cite{mislove-2007-socialnetworks} & 3.07M & 117.18M & Social\\
    \hline
    Wiki-topcats\cite{DBLP:conf/kdd/YinBLG17} & 1.79M & 28.51M & Network\\
    \hline
    LiveJournal\cite{DBLP:conf/kdd/BackstromHKL06} & 4.84M & 68.99M & Social\\
    \hline
    WRN\cite{nr} & 23.9M & 28.9M & Road\\
    \hline
    Twitter\cite{DBLP:conf/www/BoldiV04, DBLP:conf/www/BoldiRSV11} & 41.65M & 1.468B & Social\\
    \hline
    UK-2007-02\cite{url:UK-2007-02} & 110.1M & 3.945B & Social\\
    \hline
  \end{tabular}
\end{table}

For testing the scalability, we build a GPU cluster with 6 physical nodes, each of which is equipped with CPU Xeon E5-2698 v4 ($2.20$GHz, $20$ cores) and $2$ NVIDIA V$100$ GPUs ($16$GB GPU memory on each GPU).
Other experiments are run at a NVIDIA DGX workstation with CPU E5-2698 ($2.20$GHz, $20$ cores), and $4$ NVIDIA V$100$ GPUs.
For the middleware accelerator abstraction, we treat CPU in one node as an accelerator which has a 20-thread multithread processing model, and we treat each GPU as an accelerator which has 1024-thread multithread processing model.
We build our system with Ubuntu 16.04.5 LTS and deploy a Nvidia-docker framework for simulating the distributed environment with heterogeneous processors.
We construct a cross compilation solution by using maven, sbt and cmake to manage and config the global project dependencies. For Spark runtime, we use Java 8 and Scala 2.11. For local C++ and CUDA programming, g++$7$ and CUDA $10.0$ are used. The source code is available in \href{https://github.com/thoh-testarossa/GX-Plug}{GX-Plug repository}. Also, code for GraphX integration can be found at \href{https://github.com/Kamosphere/GraphXwithGPU}{GraphXwithGPU repository}\cite{DBLP:conf/wise/LiZKGX21}, and code for PowerGraph integration can be found \href{https://github.com/cave-g-f/PowerGraph-GPU}{PowerGraph-GPU repository}. 

\subsection{Results}



\subsubsection{Results on real graphs}

%
%
Figure \ref{fig:EXP_BenchmarkReal} compares the performance of GraphX and PowerGraph with non-acclerator(no prefix), CPU-integrated (prefix CPU+) and GPU-integrated (prefix GPU+) The y-axis is in log scale.
In relatively computation-dense applications such as LP and SSSP-BF, acceleration can be observed in total time.

On GraphX, compared with GraphX, GPU+GraphX achieves up to 7x acceleration in SSSP-BF, and up to 20x acceleration in LP algorithm. CPU+GraphX also achieves up to 4x acceleration in SSSP-BF algorithm and up to 5x acceleration in LP algorithm.
On PowerGraph, GPU+PowerGraph achieves up to 25x acceleration in SSSP-BF algorithm and up to 15x acceleration in LP algorithm. CPU+PowerGraph also achieves up to 5x acceleration in SSSP-BF algorithm and up to 10x acceleration in LP algorithm.
These results verify the effectiveness of the middleware.

The results on the scalability are shown in Figure~\ref{fig:EXP_BenchmarkCapacityScale}.
First, we compare three competitors, PowerGraph+GX-plug, Gunrock~\cite{DBLP:conf/ppopp/WangDPWRO16}, and Lux~\cite{DBLP:journals/pvldb/JiaKSMEA17}, by varying the number of GPUs in Figure~\ref{fig:EXP_BenchmarkCapacityScale} (a-b).
Gunrock is a single-node single-GPU graph system, and Lux is a multi-node multi-GPU graph system.
In Figure~\ref{fig:EXP_BenchmarkCapacityScale} (a), the result (Orkut, PageRank) shows that the runtime cost decreases w.r.t. the number GPUs. Gunrock performs the best on the single-GPU setting, but the multi-GPU setting is not supported. 
When there are more than $2$ GPUs, the performance of PowerGraph+GX-plug is better than Lux, and the lead is growing, showing better scalability w.r.t. the number of GPUs.
We proceed to evaluate the scalability on larger datasets, i.e., Twitter and UK-2007, in Figure~\ref{fig:EXP_BenchmarkCapacityScale} (b).
Gunrock gets overflowed on the two datasets, because it only supports the single-GPU setting and the graph data cannot be accommodated by a single GPU.
PowerGraph+GX-plug performs better than Lux on the two datasets. For example, PowerGraph+GX-plug is about $40$\% faster than Lux when processing Twitter with $4$ GPUs.
The technology pathways of Lux and GX-plug are different.
The former focuses on exploiting GPU internal mechanisms, while the latter explores more optimizations on the upper system end, e.g., synchronization skipping, which may become more critical for the scalability on large datasets.
There is no result for using $4$ GPUs on UK-2007, for all methods, because the system GPU memory capacity is exceeded. Then, we examine the scalability of PowerGraph+GX-plug on different graph algorithms in Figure~\ref{fig:EXP_BenchmarkCapacityScale} (c).
It can be observed that the sublinear speedup in computation is achieved. In particular, the runtime cost of SSSP-BF decreases from $14$s to $7$s, if the number of GPUs is increased from $2$ to $4$.
The result of mixing and matching different accelerators, i.e., CPUs and GPUs, is shown in Figure~\ref{fig:EXP_BenchmarkCapacityScale} (d).
It shows that the runtime cost decreases if the computation power increases. We will discuss the workload balancing on the heterogeneous system.

\subsubsection{Effect of Pipeline Shuffle}

Figure~\ref{fig:EXP_PipelineShuffle} shows the experiment results of the performance of pipeline shuffle mechanism. We consider $3$ competitors, ``Pipeline*'', ``Pipeline'' and ``Without pipeline'', in corresponding to the results with optimal blocksize, fixed blocksize, and the one without pipeline parallelism. Experiment shows that ``Pipeline*'' can achieve 30\%-50\% acceleration rate compared with ``withoutPipeline''. Also, ``Pipeline*'' can improve pipeline performance as 20\%-30\%, compared with ``Pipeline''.


\subsubsection{Effect of Synchronization Caching \& Skipping}

Figure~\ref{fig:EXP_SynchronizationCachingSkipping} (a) shows the performance of the synchronization caching mechanism. We use both synthetic and real graph datasets as input, and use SSSP-BF algorithm for testing the workload. The experiment result shows that we can get 2-3x acceleration in GraphX integrations. For the results on PowerGraph, it is much more efficient than GraphX. We can get up to 150\% acceleration in both synthetic and real datasets.

Figure~\ref{fig:EXP_SynchronizationCachingSkipping} (b) shows the performance of synchronization skipping mechanism. We use SSSP-BF algorithm for testing the workload, and count the number of iterations skipped on both synthetic and real datasets. We also compare the result with the number of  iterations when synchronization skipping mechanism is disabled.
For real datasets, the synchronization skipping mechanism achieves 60\%-90\% decrease of the number of iterations.
However, the effect on synthetic dataset is insignificant, where the data are more uniform, due the random generation of nodes and edges. For real datasets, there tends to be more clusters of dense partitions, leading to better partitioning results that triggers synchronization skipping.

\subsubsection{Effect of Workload Balancing}

Figure~\ref{fig:EXP_WorkloadBalancing} shows the results of workload balancing. We compare the difference of system performance with and without workload balancing. Also, we plot the best performance can be achieved in accordance to our estimation model as discussed in Section~\ref{subsec:WorkloadBalancing}.

Figure~\ref{fig:EXP_WorkloadBalancing} (a)
shows the scenario in which the hardware configuration of distributed nodes are fixed, and the partitioning strategy can be tuned (Case 1, Section~\ref{subsec:WorkloadBalancing}).
We construct two distributed nodes for the experiment. One node contains $1$ GPU + $1$ CPU, and the other contains $3$ GPUs + $1$ CPU.
We evenly partition the graph dataset to all nodes, which is the default setting of distributed graph systems, and is denoted as ``Not Balanced''.
We compare it to the one with our balancing strategy as discussed in Section~\ref{subsec:WorkloadBalancing}, which is denoted as ``Balanced''.
It shows that the workload balancing can significantly improve the system performance. Also, the balanced result is very close to the theoretically optimal result.

Figure~\ref{fig:EXP_WorkloadBalancing} (b) shows the scenario in which the partitioned results are fixed, and the hardware configuration can be tuned (Case 2, Section~\ref{subsec:WorkloadBalancing}). We construct $2$ distributed nodes with the same hardware configuration.
We vary the data load of distributed nodes to observe the effect of hardware configuration tuning.
Without balancing techniques, both distributed nodes are with $1$ GPU, denoted as ``Not Balanced''.
With balancing techniques, we can estimate the number of GPUs needed in accordance to the data load and dynamically allocate appropriate number GPUs, denoted as ``Balanced''.
It shows that the workload balancing can significantly improve the system performance. Also, the balanced result is very close to the theoretically optimal result, demonstrating the merits of workload balancing strategies.



\subsubsection{Effect of Runtime Isolation}

We hereby examine the performance of computation daemon on the runtime isolation by designing a comparative test to compare the influence of GPU initialization between daemon-agent based solution and direct GPU call solution. A larger number of iterations corresponds to a higher number of times of CPU-GPU runtime environment switching.
Results in Figure~\ref{fig:EXP_GPUInit} (11 iterations) show that our solution significantly reduces unnecessary initialization costs.
The benefits would be amplified when the number of iterations is increased.

\subsubsection{Middleware Scalability}
\label{subsubsec:Middlewarecost}


We examine the scalability of our middleware, by varying the number of distributed nodes, in Figure~\ref{fig:EXP_MidCost}.
It plots the ratio of time cost taken by the middleware to the cost of the entire system, for different graph tasks on different distributed systems, e.g., PowerGraph and GraphX.
It can be observed that, for all graph tasks, the time ratio of the middware decreases w.r.t. the increase of number of distributed nodes. The downhill trend reflects good scalability of the middleware in a larger scaled distributed computing environment, where the cost can be dominated by the gradually enlarged synchronization overhead of distributed system side.
Also, the time ratios of middleware are mostly between 10\% and 20\%, especially for algorithms with high operational intensities. Particularly, PageRank takes only about 10\% of total cost in a distributed system with 32 nodes. LP is different, since it is a fully iterative algorithm, corresponding to a low operational intensity.

In summary, the low cost ratio and the downhill trend demonstrate good scalability of our middleware.

\subsubsection{Block Size Selection}
\label{subsubsec:BSizeSelection}

To examine the effect of the block size selection, we measure the pipeline performance in different with different amount of blocks $s$, in Figure~\ref{fig:EXP_PipelineShuffleBlockSizeSelection}. We also compare the estimated $s_{opt} = \frac{d}{b_{opt}}$ with the real result\footnote{Coefficients are tested as follows: for SSSP: ($k_{1}$, $k_{2}$, $k_{3}$, $a$) = (0.03, 0.51, 0.09, 84671); for PR: ($k_{1}$, $k_{2}$, $k_{3}$, $a$) = (0.02, 0.58, 0.1, 1970); for SSSP: ($k_{1}$, $k_{2}$, $k_{3}$, $a$) = (0.003, 0.59, 0.006, 498).}.
For both LP and PageRank algorithms, we use the first iteration as the testing data. For SSSP algorithm, we use $6$-th iteration as the testing data, since the computation workload is the maximum during the entire execution. We can find that when $s$ increases, iteration time cost first decreases, and then increases. Thus, for $b = \frac{d}{s}$, when $b$ increases, iteration time cost also tends to first decrease, and then increase.

We also give our estimated $s$ following the analysis here for the $3$ different algorithms. It shows that when real $b$ and $s$ are close to the estimated result, the pipeline performance is also close to the estimated one, showing the accuracy of our estimation. Also, the optimal performance can be reached when real $b$ and $s$ are close to our estimation.

\section{Related Works}

\label{sec:relate}



{\bf Distributed CPU-based Systems.} With the prosperity of distributed system, people investigate common operator sets inside diverse graph primitives for scaling out in the distributed environment.
As a forerunner, Pregel\cite{DBLP:conf/sigmod/MalewiczABDHLC10} is proposed by Google on large-scale graph computing, following the BSP model. In BSP, graph computation are divided into iterations and intermediate results can be globally synchronized at barriers called super-steps.
GraphLab\cite{DBLP:journals/pvldb/LowGKBGH12} allows asynchronous computation and dynamic asynchronous scheduling, whose programming model also isolates the user-defined algorithm from data movement.
To achieve better workload balancing on natural graphs, PowerGraph\cite{DBLP:conf/osdi/GonzalezLGBG12} uses a more flexible GAS abstraction for power-law graphs.

There are also many embedded graph processing systems built on existing distributed systems to gain the benefits of task scheduling and data management.
GraphX\cite{DBLP:conf/osdi/GonzalezXDCFS14} is one of the most successful representative built on top of Apache Spark\cite{url:Spark}.
HaLoop\cite{DBLP:journals/vldb/BuHBE12} is a similar distributed graph processing system, in particular, extended from Hadoop\cite{url:Hadoop}.

However, most works pay little attention to the computation intensiveness of large-scale graph processing.
Efforts on scheduling balancing and data accessing also incur extra cost in computation, making the system even slower than single-node solutions. It is thus desired to have a scale-up solution for distributed graph systems.

{\bf Single-node Parallel Graph Algorithms.}
There also exist hardwired graph primitive implementations for the single-node environment.
Merrill et al. propose linear parallelization of BFS algorithm on GPU \cite{DBLP:conf/ppopp/MerrillGG12}.
Soman et al. studies graph algorithms based on two PRAM connected-component \cite{DBLP:conf/spaa/Greiner94}.
Several parallel Betweenness Centrality implementations are available on GPU based on the work of Brandes et al.~\cite{brandes2001betweeness}. Davidson et al.\cite{DBLP:conf/ipps/DavidsonBGO14} propose a work-efficient Single-Source Shortest Path algorithm on GPU.

Low-level graph parallel solutions can have best performance only on specific computation tasks, but are not general for diversified graph applications.
Also, the hardwired primitives are challenging to even skilled algorithm engineers, making such solutions hard for being deployed in real systems and applications.

{\bf High-level GPU Programming Model.} There are also existing works on high-level graph operations for GPU.
Zhong and He devise Medusa~\cite{DBLP:journals/sigmod/ZhongH14} on a high-level GPU-based system for parallel graph processing using a message-passing model, which is arguably the earliest work for GPGPU development for graphs.
CuSha~\cite{DBLP:conf/hpdc/KhorasaniVGB14}, targeting a GAS abstraction, avoids non-coalesced memory accessing  and avoids irregular memory accessing.
Gunrock~\cite{DBLP:conf/ppopp/WangDPWRO16} implements a novel data-centric abstraction centered on operations on a vertex or edge frontier rather than designing an abstraction for computation.
Recently, there are a few works on GPU graph processing system built on distributed systems, among which
Lux~\cite{DBLP:journals/pvldb/JiaKSMEA17} is one of the representatives. Users can use GPUs in multiple physical nodes for efficient computation.
However, without the support of mature distributed systems, Lux faces a series of challenges, such as robust distributed data management, scheduling balancing, effective fault recovery, and efficient data synchronization with physical layers, and thus falls short in addressing technical issues arise in large-scale graph data management and analytics.

\section{Conclusion}

\label{sec:con}

In this paper, we propose a middleware for the integration of heterogenous distributed graph systems and accelerators. 
Our middleware is versatile in the sense that it supports different programming models, computation models, and runtime environments. 
For reinforcing the middleware performance, we devise a series of techniques, such as pipeline shuffle, synchronization caching and skipping, partitioning, and parameter configuration for intra-, inter, and beyond iteration optimization.
Extensive experiments show that our middleware achieves good performance in large-scale graph processing.


\bibliographystyle{IEEEtran}
\bibliography{vldb_sample}




\end{document}